\documentclass[pdftex,11pt]{article}
\usepackage{amsmath,amsfonts,amssymb,amsthm,booktabs,color,epsfig,graphicx,hyperref,url}
\usepackage[margin=30truemm]{geometry}
\usepackage{mathdots}

\theoremstyle{plain}
\newtheorem{theorem}{Theorem}

\newtheorem{proposition}{Proposition}
\newtheorem{lemma}{Lemma}

\theoremstyle{definition}
\newtheorem{definition}{Definition}
\newtheorem{remark}{Remark}
\newtheorem{example}{Example}
\usepackage{algorithmic}
\usepackage{algorithm}

\usepackage{graphicx}
\graphicspath{ {results} }
\usepackage[labelfont=bf]{caption}

\makeatletter
\newcommand{\pushright}[1]{\ifmeasuring@#1\else\omit\hfill$\displaystyle#1$\fi\ignorespaces}
\newcommand{\pushleft}[1]{\ifmeasuring@#1\else\omit$\displaystyle#1$\hfill\fi\ignorespaces}
\makeatother

\title{On the minimum information checkerboard copulas under fixed Kendall's rank correlation}

\author{Issey Sukeda\footnote{\url{sukeda-issei006@g.ecc.u-tokyo.ac.jp}. Postal address: 7-3-1 Hongo, Bunkyo-ku, Tokyo 113-8656, Japan} \hspace{1cm} Tomonari Sei\\ Graduate school of Information Science and Technology, The University of Tokyo, Japan}

\begin{document}
\maketitle

\begin{abstract}
Copulas have gained widespread popularity as statistical models to represent dependence structures between multiple variables in various applications. The minimum information copula, given a finite number of constraints in advance, emerges as the copula closest to the uniform copula when measured in Kullback-Leibler divergence.
In prior research, the focus has predominantly been on constraints related to expectations on moments, including Spearman's $\rho$. This approach allows for obtaining the copula through convex programming. However, the existing framework for minimum information copulas does not encompass non-linear constraints such as Kendall's $\tau$.
To address this limitation, we introduce MICK, a novel minimum information copula under fixed Kendall's $\tau$. We first characterize MICK by its local dependence property. Despite being defined as the solution to a non-convex optimization problem, we demonstrate that the uniqueness of this copula is guaranteed when the correlation is sufficiently small.
Additionally, we provide numerical insights into applying MICK to real financial data.
\end{abstract}

\section{Introduction\label{sec:1}}

Uncertainty modelling, aiming at identifying the true distribution in real world under limited prior knowledge, is a fundamental approach widely used in the areas of operations research and finance. Since most of real data are multivariate and dependent on each other, it is essential to include knowledge about dependencies in these distributions. Common methods to achieve this is the use of copulas~\cite{nelsen2007introduction}.

Meeuwissen and Bedford~\cite{MEEU1997} introduced a distribution for uncertainty analysis, called the minimum information copula. Let $I = \int_{[0,1]^2} p(x,y)\log{p(x,y)} \mathrm{d}x\mathrm{d}y$ be the information of a bivariate copula $p(x,y)$ defined on $[0,1]^2$, which is also the negative of Shannon entropy. The minimum information copula is defined as a copula obtained by minimizing $I$ under a constraint on its correlation $E[xy]$. 
Roughly speaking, the belief in this approach is that we should opt for the least informative distribution when the true one is misspecified by following the maximum entropy principle proposed by Edwin Jaynes~\cite{jaynes1957information}, which is fundamental in physics. 

Minimizing the information (or equivalently maximizing the Shannon entropy) of the unknown distribution has become a popular approach for uncertainty modelling. For example, the maximum entropy copula is used to analyze drought risks by Yang et al.~\cite{yang2020mcmc}, stream by Kong et al.~\cite{kong2015maximum}, and rainfall by Qian et al.~\cite{qian2018modelling} .
Also in statistical literature, the similar problems were considered, while the problem setting varies slightly. Pougaza and Mohammad-Djafari~\cite{pougaza2012new} considered the problem of maximizing Tsallis and Renyi entropies and derived new copula families. Butucea et al.~\cite{butucea2015maximum} considered the maximum entropy copula with given diagonal section. 
Piantadosi et al.~\cite{piantadosi2012copulas} considered a class called checkerboard copula. This class has a finite number of regions on which probability is uniformly distributed. 
The checkerboard copula is known to be one of the approximations of continuous copula. Through this approximation, the problem of dealing continuous copula density is changed to the problem of dealing with a checkerboard copula having a step function density, which let us apply methods for finite probability spaces. Samo~\cite{samo2021inductive} used the similar idea to estimate the mutual information.
Furthermore, the minimum information checkerboard copula with several constraints on moments of the distribution has been extensively studied by Bedford and Wilson~\cite{bedford2014construction}, and Sei~\cite{sei2021}. Overall, although the specific notions and domains of interest varies in each studies, the most natural distribution is obtained by minimizing the information under given constraints.

However, it is notable that these studies on minimum information copulas only cover the linear constraints as prior knowledge. 
Many of the previous works on minimum information copulas assume the use of rank correlations and moments estimated from real world data as the constraints given to the optimization problem. Spearman's $\rho$ and Kendall's $\tau$ are two examples of rank correlations widely known and used. While Spearman's $\rho$ is dealt with by Meeuwissen and Bedford~\cite{MEEU1997}, Kendall's $\tau$ cannot be handled in similar manner due to the fact that Kendall's $\tau$ is not a linear notion with respect to the joint density. 

\subsection{Our contribution}

In this study, we investigate the minimum information checkerboard copula under fixed Kendall's $\tau$ (MICK) from both theoretical and algorithmic aspects. This work extends the coverage of the previous minimum information copulas, especially those proposed by Meeuwissen and Bedford~\cite{MEEU1997} and Piantadosi et al.~\cite{piantadosi2012copulas}, where the similar problem setting with Spearman's $\rho$ was the main focus.
The key difference is that the optimization problem to obtain the copula unfortunately becomes non-convex in our setting, which makes it unable to apply the similar arguments from previous studies.

Specifically, we first characterise MICK from the local dependence. 
The most famous notion of local dependence is \textit{log odds ratio}, which is the discrete version of the local dependence function proposed by Holland and Wang~\cite{Holland}. We show that the slight variant of \textit{log odds ratio}, which we name \textit{pseudo log odds ratio}, plays a central role in determining MICK.
In doing so, we introduce a novel non-orthogonal basis to represent the space of checkerboard copulas.

Moreover, we guarantee the uniqueness of MICK under a condition that the dependence is not strong enough. Although it does not lead to the uniqueness of MICK in every setting, we conjecture that the result is consistent for every Kendall's $\tau$ from several reasons. In our proof, we show the convexity at every stationary point of the optimization problem for MICK by calculating the Hessian matrix. Due to the continuity of the objective function (the information $I$), the uniqueness of the stationary point and the global optimum is guaranteed~\cite{gabrielsen1986}.

Additionally, we offer numerical insights into the application of MICK to real financial data, emphasizing its practical contributions. We present a greedy method for computing MICK numerically, leveraging the unique characteristic of this copula—its constant pseudo log odds ratio. This algorithm overcomes the shortcomings of the minimum information copula in terms of tractability and computational limitations in practice.

\subsection{Paper organization}

The rest of the paper is organized as follows: In Section 2, we review the relationship between copulas and rank correlations, and describe a general minimum information copula framework. In Section 3, we present our problem setting and provide theoretical characterizations of MICK. Section 4 highlights the commonalities and differences between MICK and MICS through a parallel argument. Section 5 is devoted to applications in practice. Finally, in Section 6, we conclude the study and suggest future directions.

\section{Rank correlations and minimum information copulas \label{sec:2}}

\subsection{Copulas and checkerboard copulas}

A $d$-dimensional copula is a joint distribution function on $[0,1]^d$ with uniform marginals (see, e.g. Nelsen~\cite{nelsen2007introduction}). In this paper, we exclusively deal with bivariate (two-dimensional) absolutely continuous copulas.

\begin{definition}[Copula Density] 
A bivariate copula density is a function $c:[0,1]^2 \to [0,\infty)$ that satisfies the following properties:
\begin{align}
\int_0^1 c(x,y)\ \mathrm{d}x = 1, \forall y \in [0,1]
\end{align}
and 
\begin{align}
\int_0^1 c(x,y)\ \mathrm{d}y = 1, \forall x \in [0,1].
\end{align}
\end{definition}
As a method for discretely approximating continuous copulas, a checkerboard copula is defined almost everywhere using a step function on multiple uniform subdivisions of $[0,1]^2$. For simplicity, we assume all subdivisions to be identical square regions: $D_{ij} = (\frac{i-1}{n},\frac{i}{n})\times (\frac{j-1}{n},\frac{j}{n})$ for $i=1, \dots, n$ and $j=1, \dots, n$. 
A checkerboard copula density can be considered identical to a square matrix $P$ through the following relationships~\cite{sei2021}: 
$$p(x,y) = n^2p_{ij}\ ((x,y) \in D_{ij}), $$
where $p(x,y)$ is a checkerboard copula density and $P = (p_{ij})$ is a $n\times n$ matrix. Since these two expressions of the checkerboard copula are equivalent up to scale, we do not distinguish these two expressions and call $P$ a checkerboard copula in this paper. 

\begin{definition}[Checkerboard copula]
    A $n\times n$ \textit{checkerboard copula} is a $n\times n$ non-negative matrix $P = (p_{ij})$ such that
    $$\sum_{i=1}^n p_{ij} = \sum_{j=1}^n p_{ij} = \frac{1}{n}.$$
\end{definition}
\noindent Here, we present two extreme examples of the checkerboard copula.
\begin{example}[Uniform checkerboard copula]
A $n \times n$ uniform checkerboard copula is a matrix with size $n \times n$, where all entries have the value of $\frac{1}{n^2}$ : $P = (p_{ij}), p_{ij} = \frac{1}{n^2}\ \forall i,j \in \{1,\dots,n\}$.
\end{example}
\begin{example}[Comonotone checkerboard copula]
A $n \times n$ comonotone checkerboard copula is represented as a diagonal matrix with size $n \times n$ : $P = \mathrm{Diag}(\frac{1}{n}, \dots, \frac{1}{n})$.
\end{example}
\noindent Note that using this expression, $nP$ is a doubly stochastic matrix with all row sums and column sums being 1. In other words, there exists a one-to-one correspondence between a checkerboard copula and a doubly stochastic matrix. We refer to Durrleman et al.~\cite{durrleman2000copulas} for the definition of the checkerboard copula density using characteristic functions that is consistent with Definition 2, and Piantadosi et al.~\cite{piantadosi2012copulas} and Kuzmenko et al.~\cite{kuzmenko2020checkerboard} for the  more formal definition of checkerboard copulas.

\subsection{Dependences}

Dependence between two different variables has always been of great interest in statistics, and numerous measures have been studied to assess it. Among them, Kendall's $\tau$ and Spearman's $\rho$ are widely used due to their invariance to marginal transformations. These measures, also referred to as rank correlations, are solely determined by the copula function.

\begin{definition}[Kendall's $\tau$, Spearman's $\rho$ (e.g., Nelsen, 2006~\cite{nelsen2007introduction})]
Let $X$ and $Y$ be continuous random variables whose copula is $C$. Then, the population version of Kendall's $\tau$ and Spearman's $\rho$ for $X$ and $Y$ is given by 
$$\tau_{X,Y} = 4\int_0^1 \int_0^1 C(u,v) dC(u,v) - 1,$$
and
$$\rho_{X,Y} = 12\int_0^1 \int_0^1 C(u,v) dudv - 3,$$
respectively.
\end{definition}
Also, the checkerboard copulas version of the rank correlations was derived based on these definitions in Durrleman et al.~\cite{durrleman2000copulas}.



\begin{theorem}[Kendall's $\tau$ of checkerboard copula (Durrleman et al., Theorem 15~\cite{durrleman2000copulas})]
Let $P = (p_{ij})$ a checkerboard copula of the size $n \times n$. Let $\Xi = (\xi_{ij}) \in \mathbb{R}^{n \times n}$ where
\begin{align}\label{eq:xi}
\xi_{i,j} = 
\begin{cases}
1& \mathrm{if}\ i = j\\
2& \mathrm{if}\ i > j\\
0& \mathrm{if}\ i < j
\end{cases}
\end{align}
then Kendall's $\tau$ of the checkerboard copula $P$ is 
\begin{align}\label{eq:tau}
\tau_P = 1-\mathrm{tr}(\Xi P \Xi P^\top).    
\end{align}
\end{theorem}


\begin{theorem}[Spearman's $\rho$ of checkerboard copula (Durrleman et al., Theorem 16~\cite{durrleman2000copulas})]
Let  $P = (p_{ij})$  a checkerboard copula of the size $n \times n$. Let $\Omega=(\omega_{ij}) \in \mathbb{R}^{n \times n}$ where
$$\omega_{ij} = \frac{1}{n^2}(n-i+\frac{1}{2})(n-j+\frac{1}{2})$$
then Spearman's $\rho$ of the checkerboard copula $P$ is 
$$\rho_P = 12\left(\mathrm{tr}(\Omega P)-\frac{1}{4}\right).$$
\end{theorem}
\noindent One of the differences between these two measures is that Kendall's $\tau$ is quadratic, while Spearman's $\rho$ is linear with respect to $P$. While this difference may be considered trivial in usual applications, it becomes crucial in our specific problem setting presented in the next section.


\subsection{Minimum information copulas}
A minimum information copula is a copula that satisfies constraints while having the minimum information relative to the uniform copula. This copula is introduced by Bedford and Wilson~\cite{bedford2014construction} as a solution to the issue of under-specification, where the goal is to determine the true distribution from limited information. Chen and Sei~\cite{CHEN2023105271} specifically developed a proper score for it.
\begin{definition}[Minimum information copulas ~\cite{bedford2014construction}] 
Let $h_1(x,y), \dots, h_K(x,y)$ be given functions and $\alpha_1, \dots, \alpha_K \in \mathbb{R}$. Then, the copula density $c$ that minimizes 
$$\int_0^1 \int_0^1 c(x,y)\log{c(x,y)} \mathrm{d}x\mathrm{d}y$$
subject to
$$\int_0^1 \int_0^1 c(x,y) h_k(x,y) \mathrm{d}x\mathrm{d}y = \alpha_k \ (k = 1, \dots, K)$$
is called the minimum information copula density.
\end{definition}
\begin{remark}
    Throughout this paper, we define $0\log{0} = 0$.
\end{remark}

To operationalize the concept of the minimum information copulas, the discretized version was also considered in Bedford and Wilson~\cite{bedford2014construction}, which we refer to as \textit{minimum information checkerboard copula}.




\begin{definition}[Minimum information checkerboard copula]
Let $h_1(x,y), \dots, h_K(x,y)$ be given functions and $\alpha_1, \dots, \alpha_K \in \mathbb{R}$. Then, the checkerboard copula $P = (p_{ij})$ that minimizes 
$$\sum_{i=1}^n \sum_{j=1}^n p_{ij}\log{p_{ij}}$$
subject to
$$\sum_{i=1}^n p_{ij} = \frac{1}{n}, \ \sum_{j=1}^n p_{ij} = \frac{1}{n},$$
$$p_{ij} \geq 0,$$
$$\sum_{i=1}^n \sum_{j=1}^n p_{ij} h_{k,ij}  = \alpha_i \ (k = 1, \dots, K),$$
is called the \textit{minimum information checkerboard copula}, where $h_{k,ij}$ is the value of $h_k(x,y)$ at the centers of each grid of the checkerboard copula. Its unique solution is known to have the form
$$p_{ij} = A_i B_j \exp{\left(\sum_{k=1}^K \theta_k h_{k,ij} \right)},$$
where $A_i$ and $B_j$ are for normalization.
\end{definition}

\noindent An important example of minimum information checkerboard copulas arises when Spearman's $\rho$ is specified by a constraint, referred to as MICS for short. Equivalent problems were explored by Meeuwissen and Bedford~\cite{MEEU1997} and Piantadosi et al.~\cite{piantadosi2012copulas}.

\begin{example}[Minimum information checkerboard copula under fixed Spearman's rank correlation (MICS)]\label{ex:mics}

MICS is defined as the optimal solution of the following problem:


$$\mathrm{minimize}\ \sum_{i=1}^n \sum_{j=1}^n p_{ij}\log{p_{ij}},$$
$$\mathrm{s.t.}\ \sum_{i=1}^n p_{ij} = \frac{1}{n},\ \sum_{j=1}^n p_{ij} = \frac{1}{n},$$
$$p_{ij} \geq 0,$$
$$12\left(\mathrm{tr}(\Omega P)-\frac{1}{4}\right) = \mu, $$
where $\mu$ is a given constant. Its optimal solution is given by
\begin{align}
    p_{ij} &= A_i B_j \exp{\left(12\theta (\frac{i}{n}-\frac{1}{2n}-\frac{1}{2})(\frac{j}{n}-\frac{1}{2n}-\frac{1}{2})\right)},
\end{align}
where $A_i$ and $B_j$ are for normalization.
This problem is a finite-dimensional convex programming, where the objective function is strictly convex. Due to compactness, the optimal solution exists and is unique~\cite{bedford2014construction}. 
While these previous researches assume linear constraints, they are not equipped to handle non-linear constraints, making them less suitable for certain data samples. Therefore, we will explore a new class of copula where we fix Kendall's tau, a common example of a non-linear constraint.
\end{example}

\section{The proposed checkerboard copula}

In this section, we present our proposed checkerboard copula MICK, which stands for the ``minimum information checkerboard copulas under fixed Kendall's rank correlation''. Since the results from previous research only apply to problems with linear constraints, they do not extend to MICK. This is because Kendall's $\tau$ cannot be expressed linearly with respect to the copula. Moreover, the constraint fixing Kendall's $\tau$ to a constant is indeed non-convex, making the optimization problem less tractable.

We first state our problem setting in Section 3.1. Then, a novel non-orthogonal basis to represent the space of checkerboard copulas is introduced in Section 3.2. By using this representation, two main characteristics of MICK are shown. First, we characterise MICK by its \textit{pseudo log odds ratio}, a variant of usual log odds ratio in Section 3.3. Secondly, the uniqueness of MICK is confirmed in Section 3.4 under a weak assumption on the value of Kendall's $\tau$.

Furthermore, we provide in Section 4 a better understanding of these results through parallel discussions between MICK and MICS. Section 5 is devoted to numerical contributions.

\subsection{Problem Setting and notations}

The optimization problem for MICK, which is our main interest, is formulated as follows:

$$(\mathrm{P})\ \mathrm{Minimize}\ \sum_{i=1}^n \sum_{j=1}^n p_{ij}\log{p_{ij}}$$
$$\mathrm{s.t.}\ \sum_{i=1}^n p_{ij} = \frac{1}{n},\ \sum_{j=1}^n p_{ij} = \frac{1}{n}, $$
$$0\leq p_{ij}, $$
$$1-\mathrm{tr}(\Xi P\Xi P^{\top}) = \mu, $$
where $\mu\ (\in [0, 1-\frac{1}{n}])$ is a given constant and 
$$\Xi = \begin{pmatrix}
1&0&\dots&0\\
2&1&\ddots&\vdots\\
\vdots&\ddots&\ddots&0\\
2&\dots&2&1\\
\end{pmatrix}.$$
\begin{remark}
We only consider the case where dependence is positive without loss of generality. When $\mu < 0$, the problem can be reduced to $\mu > 0$ by sorting the columns in reverse order from $i = 1, \dots,n$ to $i = n, \dots, 1$.  
\end{remark}
Here, we introduce some notations for other matrices used in the following sections of this paper. We denote the $n \times n$ identity matrix as $E_n$, and $J_n$ as the $n\times n$ all-one matrix. Also, we use 
$$V = \begin{pmatrix}
0&-1&\dots&-1\\
1&0&\ddots&\vdots\\
\vdots&\ddots&\ddots&-1\\
1&\dots&1&0\\
\end{pmatrix} = \Xi-J_n \in \mathbb{R}^{n\times n},$$
$$W = \begin{pmatrix}
J_n&\Xi&\dots&\Xi\\
\Xi^\top&J_n&\ddots&\vdots \\
\vdots&\ddots&\ddots&\Xi\\
\Xi^\top&\dots&\Xi^\top&J_n\\
\end{pmatrix}  = -(V \otimes V ) + J_{n^2\times n^2} = \frac{1}{2}(\Xi\otimes\Xi^\top + \Xi^\top \otimes\Xi)\in \mathbb{R}^{n^2\times n^2},$$
$$A = \begin{pmatrix}
1&0&\dots&0\\
-1&1&\ddots&\vdots\\
0&-1&\ddots&0\\
\vdots&\ddots&\ddots&1\\
0&\dots&0&-1\\
\end{pmatrix} \in \mathbb{R}^{n \times (n-1)},$$
and
$$A^\dagger =\frac{1}{n}\begin{pmatrix}
n-1&-1&\dots&\dots&-1\\
n-2&n-2&-2&\dots&-2\\
\vdots&\vdots&\ddots&\ddots&\vdots\\
1&1&1&\dots&-(n-1)\\
\end{pmatrix} \in \mathbb{R}^{(n-1) \times n}, $$
where $\otimes$ denotes the Kronecker product. 
With these notations, $A^\dagger A = E_{n-1}$ and $AA^\dagger = E_n-\frac{1}{n}J_n$ hold. In other words, $A^\dagger$ is a left inverse matrix of $A$, which is obviously full-rank.

In the context of solving optimization problem (P), the existence and uniqueness of the optimal solution are of great interest. The objective function of (P) is continuous, and the feasible region is compact; thus, a global minimum point exists. However, the uniqueness of the optimal solution cannot be guaranteed immediately since the last constraint in (P) is non-convex.
For example, consider a checkerboard copula $P_{} = 
\begin{pmatrix}
\frac{1}{9}&\frac{2}{9}&0\\
\frac{1}{9}&0&\frac{2}{9}\\
\frac{1}{9}&\frac{1}{9}&\frac{1}{9}
\end{pmatrix}
$. Then, it follows that
\begin{align}
0.0925 = \tau_{\frac{1}{2}P_{}+\frac{1}{2}P_{}^\top} < \frac{1}{2}\tau_{P_{}} + \frac{1}{2}\tau_{P_{}^\top} = 0.0987, \label{eq:hanrei}
\end{align}
where $\tau_P$ denotes Kendall's $\tau$ of the checkerboard copula $P$.
Hence, the feasible region of (P) is not convex. The non-convexity of Kendall's $\tau$ can also be confirmed by representing it as a quadratic form using $W$ (or $V$) and the vectorization operator on a square matrix:$\mathrm{vec}(P) = (p_{11},p_{12},\cdots,p_{n-1,n},p_{n,n})^\top, P \in \mathbb{R}^{n\times n}$.

\begin{lemma}\label{lem:lexical}
Let $P = (p_{ij})$ be a checkerboard copula and $\mathbf{p} = \mathrm{vec}(P)$. Then,
$$\tau = 1-\mathrm{tr}(\Xi P \Xi P^\top) = 1-\mathbf{p}^\top W\mathbf{p} = \mathbf{p}^\top(V\otimes V)\mathbf{p},$$
\end{lemma}
\noindent The proof is done by tedious algebraic calculations. The Hessian matrix of Kendall's $\tau$ is $W$, and it is a non-semidefinite matrix. Due to the non-convex nature of the problem, the uniqueness of the optimal solution cannot be guaranteed immediately, in contrast to MICS in Example~\ref{ex:mics}. However, we show in Section 3.4 that under a weak assumption that Kendall's $\tau$ is small enough, the optimal solution is unique.

\subsection{Mass transfer operation and the space of checkerboard copulas}

Geometric interpretations are useful approach to understand checkerboard copulas. The space of discrete copulas was studied by Piantadosi et al.~\cite{piantadosi2012copulas} and Perrone et al.~\cite{perrone2019geometry}, for instance. Piantadosi et al.~\cite{piantadosi2012copulas} reformulated the representation of copulas using the fact that each doubly stochastic matrix is a convex combination of permutation matrices, thanks to Birkhoff--von Neumann theorem. 
Let $P_1 = [p_{1,ij}], \dots, P_{n!} = [p_{n!,ij}] \in \mathbb{R}^{n\times n}$ be the permutation matrices. The theorem states that there exists a convex combination 
$$p_{ij} = \sum_{k=1}^{n!} \alpha_k p_{k,ij},\ \mathrm{such\ that}\ \sum_{k=1}^{n!} \alpha_k = 1, \alpha_k \geq 0 .$$
\noindent This representation, however, is redundant in that the combination is not determined uniquely for each checkerboard copula. Instead, we attempt to represent checkerboard copulas using the following (non-orthogonal) basis.

Let us define a $n \times n$ matrix $T^{ij} = \mathbf{e}_i \mathbf{e}_j^\top + \mathbf{e}_{i+1}\mathbf{e}_{j+1}^\top - \mathbf{e}_i\mathbf{e}_{j+1}^\top  - \mathbf{e}_{i+1}\mathbf{e}_{j}^\top  (i,j = 1, 2, \dots, n-1)$, where $\mathbf{e}_i$ denotes $i$-th unit column vector.
In other words, $T^{ij}$ is a zero matrix, except for one $2 \times 2$ submatrix located at the s$(i,j)$-entry, $(i+1,j)$-entry, $(i,j+1)$-entry, $(i+1,j+1)$-entry, where it takes the form: $\begin{pmatrix}
    1&-1\\
    -1&1\\
\end{pmatrix}$. Then, the space of checkerboard copula can be expressed as a subspace of a $(n-1)^2$ dimensional vector space equipped with non-orthogonal basis $\{T^{ij}\}$, i.e., there exist unique real numbers $\{p'_{ij}\}$ such that
\begin{align}
P = U + \sum_{i=1}^{n-1} \sum_{j=1}^{n-1} p'_{ij} T^{ij}, \label{eq:basis}
\end{align}
where $U(=\frac{1}{n^2}J_n)$ denotes a uniform checkerboard copula, serving as an origin point of this space. In this context, the symbol $'$ is used solely to denote the new coordinates resulting from the change in basis and is not used to represent derivatives.
It is also convenient to introduce a vector space corresponding to checkerboard copulas. 
Let $\mathbf{p} = \mathrm{vec}(P)$, $\mathbf{p}' = (p'_{1,1}, p'_{1,2}, \dots, p'_{n-1,n-2}, p'_{n-1,n-1})^\top$, and $\mathbf{t}^{ij} = \mathrm{vec}(T^{ij})$. Then, Equation~\eqref{eq:basis} can be rewritten as 
\begin{align}\label{eq:basis-vector}
\mathbf{p} = \frac{1}{n^2}\mathbf{1}_{n^2} + \sum_{i=1}^{n-1} \sum_{j=1}^{n-1} p'_{ij} \mathbf{t}^{ij} = \frac{1}{n^2}\mathbf{1}_{n^2} + (A\otimes A)\mathbf{p}'.
\end{align}
Note that to ensure that $P$ is a checkerboard copula, there are implicit constraints on $p'_{ij}$ to prevent any entry in $P$ from becoming negative, although explicitly expressing these constraints is challenging. By rearranging Equation~\eqref{eq:basis-vector}, it is possible to obtain $(p'_{ij})$s from $(p_{ij})$s:
$$\mathbf{p}' = (A^\dag \otimes A^\dag) (\mathbf{p}-\frac{1}{n^2} \mathbf{1}_{n^2}) = (A^\dag \otimes A^\dag) \mathbf{p},$$
where $\otimes$ denotes the Kronecker product.
We provide two examples to understand this new representation.
\begin{example}[Uniform copula]
$p'_{ij} = 0\ (\forall i,j)$ indicates the uniform copula $P = \frac{1}{n^2}J_n$.
\end{example}

\begin{example}[Comonotone checkerboard copula]
$p'_{ij} = \frac{1}{n^2}\min{(i,j)} (n-\max{(i,j)}) = \frac{\min{(i,j)}}{n}(1-\frac{\max{(i,j)}}{n})$ indicates the $n\times n$ comonotone checkerboard copula $P = \frac{1}{n}E_n$.
\end{example}


Now, we provide an intuitive interpretation of the basis $\{T^{ij}\}$ from a different perspective. From the definition of $T^{ij}$, increasing the coordinate $p'_{ij}$ in equation (\ref{eq:basis}) means to first choose any $2 \times 2$ region on a checkerboard copula $P$ and then transfer probability mass from two anti-diagonal entries to the other two diagonal entries. In other words, its diagonal regions increase by $\Delta (>0)$ while its anti-diagonal regions increase by $-\Delta$, keeping its row sum and column sums the same. This movement guarantees that $P$ is still a checkerboard copula after the transfer. 
When you start from a uniform copula and try applying it as many times as possible on a copula, you will eventually arrive at the co-monotone copula.
With this new basis $\{T^{ij}\}$, the space of checkerboard copulas can be visualized for better understanding of their properties. 
The space of discrete $I \times J$ bivariate copulas is known to correspond to a polytope called \textit{generalized Birkhoff polytope}. When $I = J = n$, it corresponds to a 
\textit{Birkhoff polytope}, noted as $\mathcal{B}_n$, which has $n!$ vertices corresponding to permutation matrices.
Therefore, the optimization problem (P) is a problem where we find a minimum information discrete distribution on intersection of $\mathcal{B}_n$ and $K$, where $K$ denotes the curve surface with a constant Kendall's rank correlation. 

Since the space of checkerboard copulas with large gridsize cannot be depicted because the degree of freedom is larger than three, we provide two examples with small degree of freedom in Figure 1 and Figure 2 for descriptive purposes. However, note that Example 6 assumes rectangle mesh grids instead of square mesh grids for the checkerboard copula.

\begin{example}[$3 \times 2$ checkerboard copulas]
Expression of a copula $P$ in new coordinates is
$$
P = 
\begin{pmatrix}
p_{11}&p_{12}\\
p_{21}&p_{22}\\
p_{31}&p_{32}
\end{pmatrix}
=
\begin{pmatrix}
\frac{1}{6}&\frac{1}{6}\\
\frac{1}{6}&\frac{1}{6}\\
\frac{1}{6}&\frac{1}{6}
\end{pmatrix}
+
p'_{11}
\begin{pmatrix}
1&-1\\
-1&1\\
0&0
\end{pmatrix}
+
p'_{21}
\begin{pmatrix}
0&0\\
1&-1\\
-1&1
\end{pmatrix}.
$$
Note that there are constraints on $p_{ij}$s : $0 \leq p_{ij}$, $\sum_{i=1}^3 p_{ij} = \frac{1}{2} (j=1,2)$ and $\sum_{j=1}^2 p_{ij} = \frac{1}{3} (i=1,2,3)$. These constraints lead to those on $p'_{ij}s$ as well : $|p'_{11}| \leq \frac{1}{6},  |p'_{21}| \leq \frac{1}{6}, |p'_{11}-p'_{21}| \leq \frac{1}{6}$.

In this example, the region where Kendall's tau equals to the constant $\mu$ is represented as 
$$1-\mathrm{tr}(\Xi P \Xi P^\top) = \mu \Leftrightarrow \frac{4}{3}p'_{11} + \frac{4}{3}p'_{21} = \mu,$$
which is depicted as a line in Figure \ref{fig:mesh1}.
\end{example}

\begin{example}[$3 \times 3$ checkerboard copulas]
Expression of a copula $P$ in new coordinates is
\begin{align*}
P = 
\begin{pmatrix}
p_{11}&p_{12}&p_{13}\\
p_{21}&p_{22}&p_{23}\\
p_{31}&p_{32}&p_{33}
\end{pmatrix}
&=
\begin{pmatrix}
\frac{1}{9}&\frac{1}{9}&\frac{1}{9}\\
\frac{1}{9}&\frac{1}{9}&\frac{1}{9}\\
\frac{1}{9}&\frac{1}{9}&\frac{1}{9}
\end{pmatrix}
+
p'_{11}
\begin{pmatrix}
1&-1&0\\
-1&1&0\\
0&0&0
\end{pmatrix}
+
p'_{21}
\begin{pmatrix}
0&0&0\\
1&-1&0\\
-1&1&0
\end{pmatrix}\\
&+
p'_{12}
\begin{pmatrix}
0&1&-1\\
0&-1&1\\
0&0&0
\end{pmatrix}
+
p'_{22}
\begin{pmatrix}
0&0&0\\
0&1&-1\\
0&-1&1
\end{pmatrix}
.
\end{align*}

\noindent The degree of freedom is four. To reduce its dimension so that it can be visualized in the 3 dimensional space, we assume symmetric matrices here temporarily, i.e., $p'_{12} = p'_{21}$. Then, we have
\begin{align*}
P &= 
\begin{pmatrix}
p_{11}&p_{12}&p_{13}\\
p_{21}&p_{22}&p_{23}\\
p_{31}&p_{32}&p_{33}
\end{pmatrix}\\
&=
\begin{pmatrix}
\frac{1}{9}&\frac{1}{9}&\frac{1}{9}\\
\frac{1}{9}&\frac{1}{9}&\frac{1}{9}\\
\frac{1}{9}&\frac{1}{9}&\frac{1}{9}
\end{pmatrix}
+
p'_{11}
\begin{pmatrix}
1&-1&0\\
-1&1&0\\
0&0&0
\end{pmatrix}
+
p'_{12}
\begin{pmatrix}
0&1&-1\\
1&-2&1\\
-1&1&0
\end{pmatrix}
+
p'_{22}
\begin{pmatrix}
0&0&0\\
0&1&-1\\
0&-1&1
\end{pmatrix}.
\end{align*}
\noindent In this example, the region where Kendall's tau equals to the constant $\mu$ is represented as 
$$2p'_{11}p'_{22} - 2p'^2_{12} + \frac{8}{9}(p'_{11} + p'_{22} + 2p'_{12}) = \mu,$$
which is depicted in Figure~\ref{fig:mesh2} as a curved surface.

\end{example}

\begin{figure}[htp]
    \centering
    \includegraphics[scale=0.5]{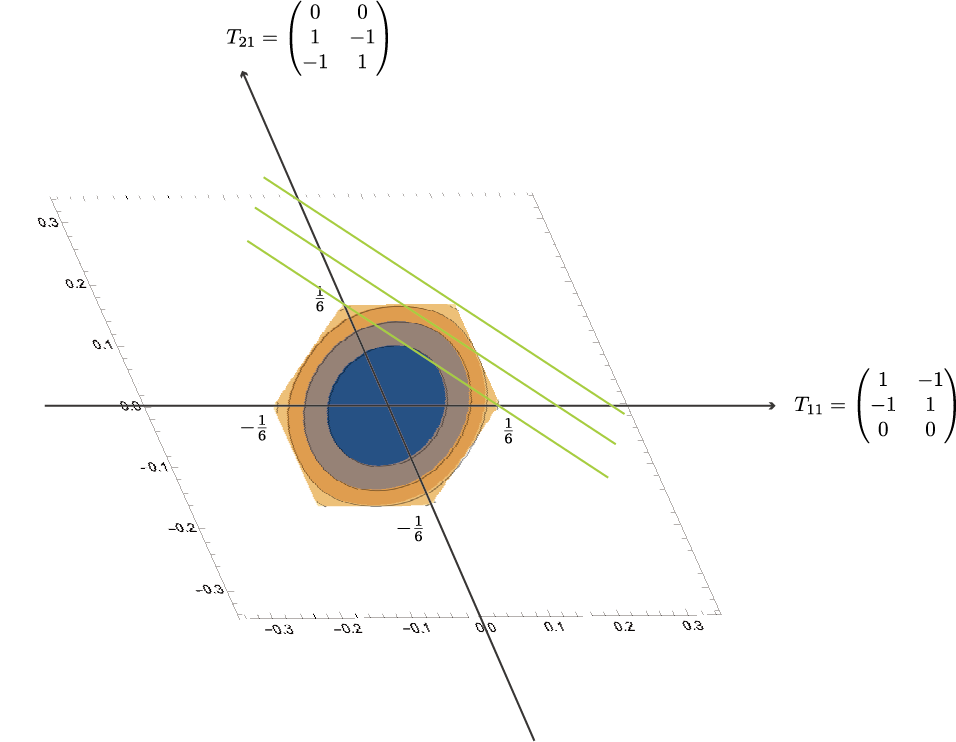}
    \caption{Visualization of $3 \times 2$ checkerboard copulas space. The irregular hexagon in the figure represents the domain of a checkerboard copula. Ovals inside it represents the contour lines of the information of copulas. The line corresponds to a region where Kendall's $\tau$ remains constant.}
    \label{fig:mesh1}
    \centering
    \includegraphics[scale=0.5]{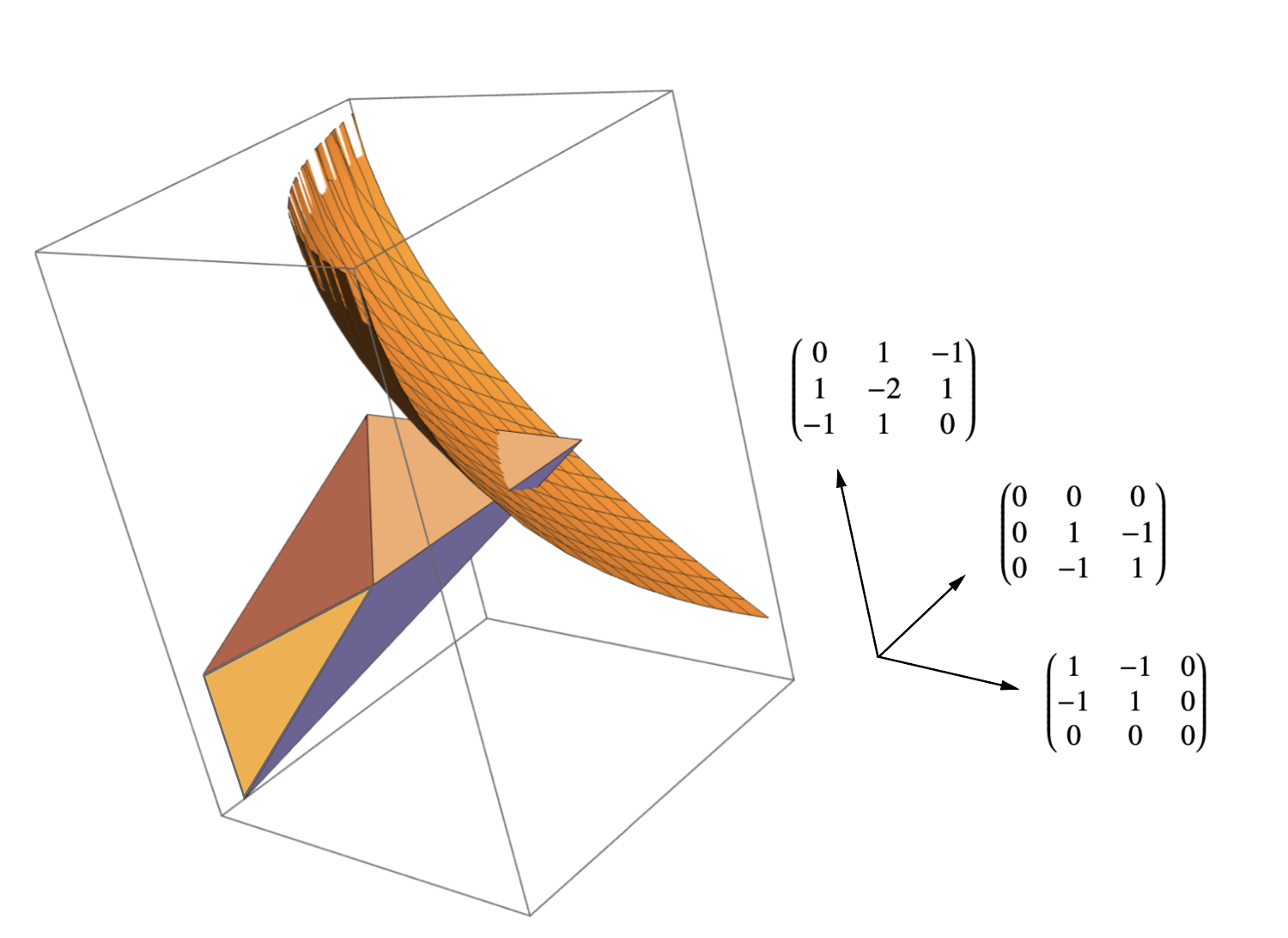}
    \caption{Visualization of $3 \times 3$ checkerboard copulas space. The polyhedron in the figure represents the domain of a checkerboard copula. The curved surface corresponds to a region where Kendall's $\tau$ remains constant.}
    \label{fig:mesh2}
\end{figure}

\subsection{Main result 1: Pseudo log odds ratio of MICK is constant everywhere}

Assume that all entries of checkerboard copulas are strictly positive. 
For every $2 \times 2$ submatrices of the matrix associated with MICK, we state that the variant of well-known odds ratio always takes a constant value. This is derived by considering the stationary conditions for the problem $(\mathrm{P})$.

\begin{lemma}[Variation of Kendall's tau]\label{lem:tau}
Let $P=(p_{ij})$ be a checkerboard copula. Consider a small change $\epsilon T^{ij} (\epsilon \in \mathbb{R})$ on $P$. The variation of Kendall's tau is 
$$\tau(P+\epsilon T^{ij}) - \tau(P) = 2\epsilon(p_{ij}+p_{i+1,j+1}+p_{i+1,j}+p_{i,j+1})+O(\epsilon^2)$$
as $\epsilon \to 0$. The value $p_{ij}+p_{i+1,j+1}+p_{i+1,j}+p_{i,j+1}$ is positive, meaning that
Kendall's tau always increases when $\epsilon$ is positive and small.
\end{lemma}
\begin{proof}[\textbf{\upshape Proof of Lemma \ref{lem:tau}:}]
\begin{align*}
(\mathbf{t}^{ij})^\top W\mathbf{p} &= (\sum_{k=1}^n \sum_{l=1}^n w_{n(i-1)+j,k} p_{kl})-(\sum_{k=1}^n \sum_{l=1}^n w_{n(i-1)+j+1,k} p_{kl})\\
&-(\sum_{k=1}^n \sum_{l=1}^n w_{ni+j,k} p_{kl})+(\sum_{k=1}^n \sum_{l=1}^n w_{ni+j+1,k} p_{kl})\\
&= \sum_{k=1}^n \sum_{l=1}^n (w_{n(i-1)+j,k}-w_{n(i-1)+j+1,k}-w_{ni+j,k}+w_{ni+j+1,k}) p_{kl}\\
&= -p_{i,j}-p_{i,j+1}-p_{i+1,j}-p_{i+1,j+1}
\end{align*}

Let $\tau_{P}$ be Kendall's tau of $P$ and $\tau_{P+\epsilon T^{ij}}$ be that of $P+\epsilon T^{ij}$. It follows that
\begin{align*}
\tau(P+\epsilon T^{ij}) - \tau(P) 
&= (1-(\mathbf{p}+\epsilon \mathbf{t}^{ij})^\top W(\mathbf{p}+\epsilon \mathbf{t}^{ij}))-(1-\mathbf{p}^\top W\mathbf{p})\\
&= \mathbf{p}^\top W\mathbf{p} - (\mathbf{p}+\epsilon \mathbf{t}^{ij})^\top W(\mathbf{p}+\epsilon \mathbf{t}^{ij})\\
&= -2\epsilon (\mathbf{t}^{ij})^\top W \mathbf{p} - \epsilon^2 (\mathbf{t}^{ij})^\top W\mathbf{t}^{ij}\\
&= 2\epsilon(p_{ij}+p_{i,j+1}+p_{i+1,j}+p_{i+1,j+1}) + O(\epsilon^2).
\end{align*}
\end{proof}

\begin{lemma}[]\label{lem:delta}
Let $\eta_{ij} = p_{ij}+p_{i+1,j}+p_{i,j+1}+p_{i+1,j+1}$. 
Then, it follows that
$$1-\mathrm{Tr}(\Xi P\Xi P^\top) = 1-\mathrm{Tr}(\Xi \tilde{P}\Xi{\tilde{P}}^\top)+O(\epsilon^2)$$
as $\epsilon \to 0$, where $$\tilde{P} = P + \epsilon (T^{ij} - \frac{\eta_{ij}}{\eta_{i'j'}} T^{i'j'}).$$
\end{lemma}
\noindent This lemma states Kendall's tau is invariant under the operation $\epsilon (T^{ij} - \frac{\eta_{ij}}{\eta_{i'j'}} T^{i'j'})$, ignoring the $O(\epsilon^2)$ term.
The proof is straightforward from Lemma \ref{lem:tau}.

\begin{lemma}[Variation of the objective function]\label{lem:obj}
Let $P=(p_{ij})$ be a checkerboard copula. Consider a small change $\epsilon T^{ij} (\epsilon \ll 1)$ on $P$. The variation of the objective function in the optimization problem (P) is 
$$I(P + \epsilon T^{ij})-I(P) = \epsilon \log{\frac{p_{ij}p_{i+1,j+1}}{p_{i+1,j}p_{i,j+1}}} + O(\epsilon^2)$$
as $\epsilon \to 0$, where $I(P) = \sum_{i=1}^n \sum_{j=1}^n p_{ij}\log{p_{ij}}$.
\end{lemma}

\begin{proof}[\textbf{\upshape Proof of Lemma \ref{lem:obj}}]
$$\frac{\partial}{\partial p_{ij}}(\sum_i \sum_j p_{ij}\log{p_{ij}}) = \log{p_{ij}}+1,$$
so the increase in the information of $P$ by the operation $\epsilon T^{ij}$ is
$$\epsilon(\log{p_{ij}}+1)+\epsilon(\log{p_{i+1,j+1}}+1)-\epsilon(\log{p_{i+1,j}}+1)-\epsilon(\log{p_{i,j+1}}+1)= \epsilon\log{\frac{p_{ij}p_{i+1,j+1}}{p_{i+1,j}p_{i,j+1}}}. $$
\end{proof}

Two variations from Lemma \ref{lem:tau} and Lemma \ref{lem:obj} lead to the following main statement.
\begin{theorem}[]\label{thm:mick}
The following value is constant for every $2 \times 2$ submatrices 
$\begin{pmatrix}
p_{ij}&p_{i, j+1}\\
p_{i+1, j}&p_{i+1, j+1}\\
\end{pmatrix}$
on an $n\times n$ MICK:
$$\frac{1}{p_{ij}+p_{i+1,j}+p_{i+1,j}+p_{i+1,j+1}}\log{\frac{p_{ij}p_{i+1,j+1}}{p_{i+1,j}p_{i,j+1}}},$$
where $i,j = 1,\dots,n-1$.
We name this common value ``pseudo log odds ratio'' of MICK.
\end{theorem}
\begin{proof}[\textbf{\upshape Proof of Theorem \ref{thm:mick}}]
The tangent space of the surface $\tau(P)=\mu$ is spanned by $T^{ij} - \frac{\eta_{ij}}{\eta_{i'j'}}T^{i'j'}$ for all $i, j, i', j'$. The stationary condition for the optimization problem (P) is 
$$\log{\frac{p_{ij}p_{i+1,j+1}}{p_{i+1,j}p_{i,j+1}}} -\frac{\eta_{ij}}{\eta_{i'j'}} \log{\frac{p_{i',j'}p_{i'+1,j'+1}}{p_{i'+1,j'}p_{i',j'+1}}} = 0,$$
where $\eta_{ij}$ is defined in Lemma~\ref{lem:delta}. By rearanging this equation, we obtain
$$\frac{p_{i',j'}+p_{i'+1,j'}+p_{i',j'+1}+p_{i'+1,j'+1}}{p_{ij}+p_{i+1,j}+p_{i,j+1}+p_{i+1,j+1}}\log{\frac{p_{ij}p_{i+1,j+1}}{p_{i+1,j}p_{i,j+1}}} - \log{\frac{p_{i',j'}p_{i'+1,j'+1}}{p_{i'+1,j'}p_{i',j'+1}}} = 0$$
and
$$\frac{1}{p_{ij}+p_{i+1,j}+p_{i,j+1}+p_{i+1,j+1}}\log{\frac{p_{ij}p_{i+1,j+1}}{p_{i+1,j}p_{i,j+1}}} =  \frac{1}{p_{i',j'}+p_{i'+1,j'}+p_{i',j'+1}+p_{i'+1,j'+1}}\log{\frac{p_{i',j'}p_{i'+1,j'+1}}{p_{i'+1,j'}p_{i',j'+1}}}$$
\noindent for any pairs $(i,j,i',j')$. 
For the last equation, note that both sides of the equation has the same form for $(i,j)$ and $(i',j')$.
\end{proof}
The interpretation of ``pseudo log odds ratio'' will be provided in Section 4 along with the comparison between MICK and MICS.

\subsection{Main result 2: MICK exists uniquely when Kendall's $\tau$ is small}

As Theorem~\ref{thm:mick} only states the sufficient condition of the optimal solution of $(\mathrm{P})$, we are naturally interested in the uniqueness of it. As our main result, we state that when the given Kendall's rank correlation (or the corresponding pseudo log odds ratio) is small enough, the optimal solution becomes unique despite the non-convexity of the problem. 

\begin{theorem}\label{thm:uniqueness}
Consider all possible Lagrange multipliers corresponding to the last constraint of (P).
    The optimization problem $(\mathrm{P})$ has a unique optimal solution when all the Lagrange multipliers are less than 2.
\end{theorem}
\noindent The proof follows by considering the information on an arbitrary curve passing through a stationary point and calculating the Hessian there. Local convexity of the information on every stationary points leads to the statement~\cite{gabrielsen1986}. The complete proof is given in~\ref{appendix:uniqueness}.

The following propositions support the assertion that the Lagrange multiplier becomes small when $\mu = \tau(P)$ is set small, indicating the unique existence of MICK when Kendall's $\tau$ is sufficiently small. Here, we show that $\lambda$ becomes small when $\mu$ is small from the stationary condition
\begin{equation}\label{eq:stationary}
    \log{p_{ij}} + 1 + \lambda (Wp)_{ij} - \alpha_j - \beta_i = 0, \lambda, \alpha_j, \beta_i \in \mathbb{R}
\end{equation}
and the constraint on Kendall's $\tau$, $\tau(P) = 1-\mathrm{tr}(\Xi P \Xi P^\top) = \mu$. Specifically, we show that $0<\lambda<2$ when $\mu < n^{-6}$. We leave it as future work to investigate whether the upper bound of $\mu$ can be chosen independently of the gridsize $n$.

\begin{lemma}\label{lem:checkerboard_copula}
    Let $P = (p_{ij})$ a checkerboard copula. Then, there exist $i < i'$ and $j < j'$ such that at least one of the following conditions holds: \\
    (i) $p_{ij} \geq n^{-2}$ and $p_{i'j'} \geq n^{-2}$ \\
    (ii) $p_{i'j} \geq n^{-2}$ and $p_{ij'} \geq n^{-2}$
\end{lemma}
\begin{proof}
    Since the sum of $P$ is 1, there exists at least one pair $i, j$ such that $p_{ij} \geq n^{-2}$. Now, assuming that for all $i' \neq i$ and $j' \neq j$, we have $p_{i'j'} < n^{-2}$, from the conditions of copulas, it follows that for $j' \neq j$, $p_{ij'} = n^{-1} -  \sum_{i' \neq i} p_{i'j'} > n^{-2}$. Furthermore, $p_{ij} + \sum_{j' \neq j} p_{ij'} > n^{-1}$, which contradicts the copula condition. Therefore, there exists some $i' \neq i$ and $j' \neq j$ such that $p_{i'j'} \geq n^{-2}$. The indices can be rearranged accordingly.
\end{proof}

\begin{proposition}
    $0<\lambda<2$ when $\mu < n^{-6}$.
\end{proposition}
First, Kendall's $\tau$ can be rewritten as
\begin{equation}\label{eq:kendalls_tau}
    \tau(P) = p^\top (V \otimes V)p = 2\sum_{i<i'} \sum_{j<j'} (p_{ij}p_{i'j'}-p_{ij'}p_{i'j}).
\end{equation}
Next, from Lemma~\ref{lem:tau} and \eqref{eq:stationary} for any $i < i'$ and $j < j'$, we have 
\begin{align}
\log{\frac{p_{ij}p_{i'j'}}{p_{ij'}p_{i'j}}} &= -\lambda\{(Wp)_{ij} + (Wp)_{i'j'} + (Wp)_{ij'} + (Wp)_{i'j} \}\\
&= \lambda \sum_{k=i}^{i'-1} \sum_{l=j}^{j'-1} (p_{kl} + p_{k+1,l} + p_{k,l+1} + p_{k+1,l+1}). \label{eq:logodds_and_lambda}
\end{align}
Therefore, each term in \eqref{eq:kendalls_tau} becomes
$$p_{ij}p_{i'j'}-p_{ij'}p_{i'j} = (p_{ij}p_{i'j'}+p_{ij'}p_{i'j})\frac{e^{\lambda r(i,j,i',j')}-1}{e^{\lambda r(i,j,i',j')}+1},$$
where we define $r(i,j,i',j') = \sum_{k=i}^{i'-1} \sum_{l=j}^{j'-1} (p_{kl} + p_{k+1,l} + p_{k,l+1} + p_{k+1,l+1})$.

Here we show $\tau(P) \geq n^{-6}$ when $\lambda \geq 2$. It suffices to show that there exist $i<i'$ and $j<j'$ such that $p_{ij}p_{i'j'}-p_{ij'}p_{i'j} > n^{-6}$. From Lemma~\ref{lem:checkerboard_copula}, there exists $i<i'$ and $j<j'$ such that $p_{ij}p_{i'j'}+p_{ij'}p_{i'j} \geq n^{-4}$. Moreover,
\begin{align*}
    e^{\lambda r(i,j,i',j')} &\geq 1 + \lambda r(i,j,i',j') \\
    &\geq 1 + 2(p_{ij} + p_{i'j'}+p_{ij'}+p_{i'j}) \\
    &\geq 1 + \frac{4}{n^2}.
\end{align*}
Therefore, 
\begin{align*}
p_{ij}p_{i'j'}-p_{ij'}p_{i'j} &= (p_{ij}p_{i'j'}+p_{ij'}p_{i'j})\frac{e^{\lambda r(i,j,i',j')}-1}{e^{\lambda r(i,j,i',j')}+1}\\
&= (p_{ij}p_{i'j'}+p_{ij'}p_{i'j})(1-\frac{2}{e^{\lambda r(i,j,i',j')}+1})\\
&\geq n^{-4} (1-\frac{2}{(1 + \frac{4}{n^2})+1})\\
&\geq n^{-6}.
\end{align*}

\begin{proposition}
    $\mu$ increases monotonically with respect to $\lambda$.
\end{proposition}
\begin{proof}
    From the mean value theorem, there exists a constant $\zeta$ such that 
    $$\log{\frac{p_{ij}p_{i'j'}}{p_{ij'}p_{i'j}}} = \log{p_{ij}p_{i'j'}} - \log{p_{ij'}p_{i'j}} = \frac{1}{\zeta}(p_{ij}p_{i'j'}-p_{ij'}p_{i'j}),$$
    $$0 < \min{(p_{ij}p_{i'j'},p_{ij'}p_{i'j})} \leq \zeta \leq \max{(p_{ij}p_{i'j'},p_{ij'}p_{i'j})}.$$
    Therefore, it follows from \eqref{eq:logodds_and_lambda} that for any $i < i'$ and $j < j'$, 
    $$p_{ij}p_{i'j'}-p_{ij'}p_{i'j} = \zeta\lambda r(i,j,i',j'), $$
    where $\zeta$ and $r(i,j,i',j')$ are always positive. Hence, $\tau(P) = \mu$, represented as in \eqref{eq:kendalls_tau}, is monotonically increasing with respect to  $\lambda$.
\end{proof}



\section{The comparison between MICK and MICS\label{sec:comparison}}

The argument in the previous section applies not only to MICK but also to MICS. Here, we draw parallel lines between our MICK and MICS from previous studies.
Let us review the optimization problem for MICS, mentioned in Example \ref{ex:mics}.
$$\mathrm{Minimize}\ \sum_{i=1}^n \sum_{j=1}^n p_{ij}\log{p_{ij}}$$
$$\mathrm{s.t.}\ \sum_{i=1}^n p_{ij} = \frac{1}{n},\ \sum_{j=1}^n p_{ij} = \frac{1}{n}, $$
$$0\leq p_{ij}, $$
$$\sum_{i=1}^n \sum_{j=1}^n h_{ij}p_{ij} = \mu, \ \mathrm{where}\ h_{ij} = 12\left(\frac{i}{n}-\frac{1}{2n}-\frac{1}{2}\right)\left(\frac{j}{n}-\frac{1}{2n}-\frac{1}{2}\right).$$
The only difference between the problem settings of MICK and that of MICS is the last constraint; MICK fixes Kendall's $\tau$ but MICS fixes Spearman's $\rho$. However, this difference is not trivial because the convexity of the optimization problem differs. While Kendall's $\tau$ becomes a non-convex constraint, Spearman's $\rho$ becomes a convex constraint. In this section, we demonstrate that while MICK and MICS are similar models, they exhibit distinct properties from various perspectives, including optimal solutions, odds ratios, and total positivity. Results are summarized in Table~\ref{tab:comparison}.

\subsection{Optimal solutions}

Unfortunately, the Lagrangian method becomes intractable for MICK. 
Differing from MICK, MICS is more tractable due to the linearity of Spearman's $\rho$ and the convexity of its optimization problem. Specifically, MICS is known to exist uniquely and is represented in a manner reminiscent of the exponential family: $p_{ij} = A_i B_j \exp{(\theta h_{ij})}$, where $h_{ij} = 12(\frac{i}{n}-\frac{1}{2n}-\frac{1}{2})(\frac{j}{n}-\frac{1}{2n}-\frac{1}{2})$. Here, $A_i$ and $B_j$ are intractable normalization functions and $\theta$ is the natural parameter corresponding to the Lagrangian multiplier of the optimization problem above. See Meeuwissen and Bedford~\cite{MEEU1997} and Bedford and Wilson~\cite{bedford2014construction} for more information of MICS.


\subsection{Log odds ratio and pseudo log odds ratio}

In Section 3, we derived the stationary condition for MICK and obtained the characterization of MICK through the ``pseudo log odds ratio''. In doing so, we compared the variation of Kendall's $\tau$ and the variation of the objective function, which is the information of the checkerboard copula with respect to the infinitesimal movement using $T^{ij}$. This flow of arguments also applies to MICS. Since the difference in Spearman's $\rho$ between checkerboard copulas $P$ and $P + \epsilon T^{ij}$ is calculated as 
\begin{align*}
    \rho(P + \epsilon T^{ij}) - \rho(P)
    &= 12(\mathrm{tr}(\Omega (P + \epsilon T^{ij}) - \frac{1}{4}) - 12(\mathrm{tr}(\Omega P) - \frac{1}{4})\\
    &= 12 \epsilon \mathrm{tr}(\Omega T^{ij})\\
    &= 12 \epsilon \frac{1}{n^2}
\end{align*}
for any $\epsilon\ (>0)$. 
The last equality follows from the definition of the matrix $\Omega$ : $\omega_{i,j} = \frac{1}{n^2}(n-i+\frac{1}{2})(n-j+\frac{1}{2})$. Therefore, $\rho_{P + \epsilon T^{ij} - \epsilon T^{i'j'}} = \rho_{P}$ for different $(i,j)$ and $(i',j')$. From the stationary condition, we have $2\epsilon \log{\frac{p_{ij}p_{i+1,j+1}}{p_{i+1,j}p_{i,j+1}}} - 2\epsilon \log{\frac{p_{i',j'}p_{i'+1,j'+1}}{p_{i'+1,j'}p_{i',j'+1}}} = 0$, meaning that \textit{log odds ratio} $\log{\frac{p_{ij}p_{i+1,j+1}}{p_{i+1,j}p_{i,j+1}}}$ is constant for every $2\times 2$ submatrices of MICS.
This result is consistent with the known formulation $p_{ij} = A_i B_j \exp{\left(12\theta (\frac{i}{n}-\frac{1}{2n}-\frac{1}{2})(\frac{j}{n}-\frac{1}{2n}-\frac{1}{2})\right)}$: log odds ratio of MICS is calculated as 
\begin{align}
    \log{\frac{p_{ij}p_{i+1,j+1}}{p_{i+1,j}p_{i,j+1}}} &= \frac{12}{n^2}
    \theta.
\end{align}
Here, note that the intractable normalizing functions $A_i$, $A_{i+1}$, $B_j$, $B_{j+1}$ cancel out during the calculation.

On the other hand, pseudo log odds ratio $\frac{1}{p_{ij}+p_{i+1,j+1}+p_{i+1,j}+p_{i,j+1}}\log{\frac{p_{ij}p_{i+1,j+1}}{p_{i+1,j}p_{i,j+1}}}$ is constant for every $2\times 2$ submatrices of MICK. In a different perspective, when we specify one point on MICK, log odds ratio around that point is proportional to the sum of probability mass around that point. Compared to MICS where log odds ratio is constant everywhere, it can be interpreted that MICK puts more mass on regions with stronger positive dependence.

In summary, MICK is defined by Kendall's $\tau$, but it can also be characterized by the pseudo log odds ratio. Similarly, MICS is defined by Spearman's $\rho$, but it can alternatively be specified by the log odds ratio. In essence, these ratios can be viewed as model parameters, with the strength of positive dependence monotonically increasing with respect to these parameters.

\subsection{Total positivity}

Finally, we argue that MICK possesses preferable dependence properties known as ``total positivity''. 
A function of two variables $f$ is said to be TP2 (short for ``total positivity of order two'') when $f(x,y)f(x',y') \geq f(x,y')f(x',y)$ for any pairs $(x,y)$ and $(x',y')$, where $x < x'$ and $y < y'$. 
Moreover, total positivity of higher order is defined as follows.
\begin{definition}[Total positivity~\cite{MEEU1997}]
    A density $f(x,y)$ is called \textit{totally positive of degree $n$} (TP$n$) if and only if for all $x_1 \leq x_2 \leq \cdots \leq x_n$ and for all $y_1 \leq y_2 \leq \cdots \leq y_n$ matrix $M$ with elements $m_{ij} = f(x_i,y_j)$ obeys $\mathrm{det}(M) \geq 0$.
\end{definition}
\noindent Recently, Fuchs and Tschimpke~\cite{FUCHS2023126629} studied total positivity of copulas. Here, we show that MICK is TP2.
To show it, we first introduce the relaxed problem of (P). By relaxing the last equality constraint in (P), we obtain the following problem (RP): 
$$(\mathrm{RP})\ \mathrm{Minimize}\ \sum_{i=1}^n \sum_{j=1}^n p_{ij}\log{p_{ij}}$$
$$\mathrm{s.t.}\  \sum_{i=1}^n p_{ij} = \frac{1}{n},\ \sum_{j=1}^n p_{ij} = \frac{1}{n},$$
$$0\leq p_{ij},$$
$$1-\mathrm{tr}(\Xi P\Xi P^{\top}) \geq \mu\ (\mu \geq 0).$$
Note that relaxing in the opposite direction is invalid, as it always results in the uniform copula. The fact that this relaxation problem (RP) is non-convex as well is evident from Example~\ref{eq:hanrei}.

\begin{lemma}~\label{lem:nonzero}
    All entries of the optimal solution of (RP) are strictly positive.
\end{lemma}
\begin{proof}
    Let $P^*$ denote the optimal solution of (RP). Assume $p_{i',j'}^* = 0\ (i', j' \in \{1,\dots,n-1\})$ for a single entry of $P^*$. Here we show this contradicts to the optimality of $P^*$. Consider a different checkerboard copula $P'$ = $P^* + \epsilon T_{i',j'}$, where $0 < \epsilon <  \frac{p^*_{i'+1,j'}p^*_{i',j'+1}}{p^*_{i'+1,j'+1}}$. By Lemma~\ref{lem:tau}, we have
$\tau_{P'} >  \tau_{P^*} = \mu$.
Hence, $P'$ belongs to the feasible set of (RP). Finally, we show that $P'$ takes the better optimal value than $P^*$. The optimal value for the solution $P$ is calculated as $I(P) = \sum_{i=1}^n\sum_{j=1}^n p_{ij}\log{p_{ij}}$. Following Lemma~\ref{lem:obj}, the difference in optimal values is
\begin{align*}
    I(P') - I(P^*) &= (p^*_{i',j'}+\epsilon)\log{(p^*_{i',j'}+\epsilon)} - \mathbf{p}^*_{i',j'}\log{\mathbf{p}^*_{i',j'}}\\
    &+ (p^*_{i'+1,j'+1}+\epsilon)\log{(p^*_{i'+1,j'+1}+\epsilon)} - p^*_{i'+1,j'+1}\log{p^*_{i'+1,j'+1}}\\
    &- \{(p^*_{i'+1,j'}+\epsilon)\log{(p^*_{i'+1,j'}+\epsilon)} - p^*_{i'+1,j'}\log{p^*_{i'+1,j'}}\}\\
    &- \{(p^*_{i',j'+1}+\epsilon)\log{(p^*_{i',j'+1}+\epsilon)} - p^*_{i',j'+1}\log{p^*_{i',j'+1}}\}\\
    &= \epsilon \log{\epsilon} + \epsilon \log{\frac{p^*_{i'+1,j'+1}}{p^*_{i'+1,j'}p^*_{i',j'+1}}}\\
    &< 0.
\end{align*}
Since (RP) is a minimization problem, this contradicts to the optimality of $P^*$.
\end{proof}

\begin{lemma}\label{lem:inner}
    The optimal solution of (RP) is not a relatively inner point.
\end{lemma}
\begin{proof}
    The objective function is convex on the set $\{P=(p_{ij})\ |\ \sum_j p_{ij} = \frac{1}{n}, \sum_i p_{ij} = \frac{1}{n}, p_{ij}\geq 0\}$ and its stationary point is uniquely $P = \frac{1}{n^2}J_n$, where $J_n$ is a $n \times n$ all-one matrix. Assume that the global optimal of (RP) $P^*$ is relatively inner point. Then, $P^*$ is a stationary point, thus $P^* = \frac{1}{n^2}J_n$. However, $\frac{1}{n^2}J_n$ is clearly not in feasible set of $(\mathrm{RP})$, which leads to the contradiction. 
\end{proof}

\begin{lemma}\label{lem:p=rp}
    The optimal solution of (P) and that of (RP) is equal.
\end{lemma}
\begin{proof}
    It suffices to show that the global optimal of (RP) belongs to the feasible region of (P), since (RP) is the relaxation of (P). 
    Let the global optimal point of (RP) be $P^* = (p_{ij})\ (i, j = 1, \dots, n$). In Lemma~\ref{lem:inner}, it is shown that $P^*$ is not a relatively inner point, meaning that $P^*$ must satisfy either or both of the following: (i) $p_{ij}^* = 0$ for some $(i,j)$, (ii) $1-\mathrm{Tr}(\Xi P^*\Xi(P^*)^\top) = \mu$.
    However, (i) is excluded due to Lemma~\ref{lem:nonzero}.
\end{proof}

These lemmas guarantee that MICK is TP2. 

\begin{proposition}
    MICK is TP2. 
\end{proposition}
\begin{proof}
Let a copula $P^* = (p^*_{ij})$ be a MICK, i.e., the optimal solution of (P). 
From Lemma~\ref{lem:nonzero}, all entries are strictly positive.
Moreover,  $P^*$ is also the optimal solution of (RP) from Lemma~\ref{lem:p=rp}. 
Here, we assume $P^*$ is not d-TP2 and show contradiction. For some $i,j \in \{1,\dots,n-1\}$, there exists a $2 \times 2$ submatrix $\begin{pmatrix}
    p^*_{i,j}&p^*_{i,j+1}\\
    p^*_{i+1,j}&p^*_{i+1,j+1}\\
\end{pmatrix}$ of $P^*$ such that its determinant is  strictly negative: $p^*_{i,j}p^*_{i+1,j+1}-p^*_{i+1,j}p^*_{i,j+1} < 0$. Then, since Lemma~\ref{lem:tau} and  
$\log{\frac{p^*_{i,j}p^*_{i+1,j+1}}{p^*_{i+1, j}p^*_{i, j+1}}} < 0$ hold, $P^* + \epsilon T^{ij}\ (\epsilon \ll 1)$ achieves smaller objective value while remaining inside the feasible region of (RP). This contradicts with the optimality of $P^*$.
\end{proof}

\begin{remark}
    The density function of MICS is TP2 as well. It is further known to be TP$n$ for any $n \in \mathbb{N}$, i.e., TP$\infty$ ~\cite{MEEU1997}.
\end{remark}

\begin{table}[h]
    \centering
    \caption{Comparison between MICK and MICS}
    {\renewcommand\arraystretch{2.5}
    \scalebox{0.8}{
    \begin{tabular}{c|c|c|}
        &MICK & MICS \\ \hline
        Objective function & Information & Information \\ \hline
        Constraint&Kendall's $\tau$&Spearman's $\rho$\\ \hline
        Variation by adding $\epsilon T^{ij} (\epsilon \ll 1)$&$2(p_{i,j} + p_{i+1,j} + p_{i,j+1} + p_{i+1,j+1})\epsilon+O(\epsilon^2)$ & $\frac{12}{n^2}\epsilon$\\ \hline
        Convexity of the problem&non-convex&convex\\ \hline
        The optimal solution $p_{i,j}$&A closed form is not known&$\propto \exp{\left(12\theta (\frac{i}{n}-\frac{1}{2n}-\frac{1}{2})(\frac{j}{n}-\frac{1}{2n}-\frac{1}{2})\right)}$\\ \hline
        Total positivity of density&TP2&TP$\infty$\\ \hline
        Constant value&\begin{tabular}{c}
        pseudo log odds ratio\\$\frac{1}{p_{ij}+p_{i+1,j+1}+p_{i+1,j}+p_{i,j+1}}\log{\frac{p_{ij}p_{i+1,j+1}}{p_{i+1,j}p_{i,j+1}}}$
        \end{tabular}&\begin{tabular}{c}
             log odds ratio\\ 
             $\log{\frac{p_{ij}p_{i+1,j+1}}{p_{i+1,j}p_{i,j+1}}} = \frac{12}{n^2}\theta$
        \end{tabular}\\ \hline
    \end{tabular}
    }
    }
    \label{tab:comparison}
\end{table}

\section{Application}

The majority of previous works on the minimum information copula were devoted to theoretical developments. On the other hand, the practical use of this specific copula has not been sufficiently studied. In this section, we demonstrate the application of MICK and MICS.

\subsection{Numerical calculation}
To utilize MICK for real data analysis, one needs to be able to compute it numerically. One naive approach is to use numerical solvers to directly obtain optimal solutions for $(\mathrm{P})$, or $(\mathrm{RP})$ since the solutions are identical as in Lemma~\ref{lem:p=rp}. However, computational limitations become problematic. Empirically, it is difficult to reach the optimal solution due to the iteration limit when the gridsize $n$ is larger than 10. 
Instead, we can take advantage of the results in Theorem \ref{thm:mick}.
We have already confirmed that pseudo log odds ratio is constant on MICK, thus the numerical solution can be obtained as the following Algorithm~\ref{alg1}:

\begin{algorithm}[H]
    \caption{Greedy calculation of MICK}
    \label{alg1}
    \begin{algorithmic}[1]
    \REQUIRE pseudo log odds ratio $r$
    \STATE $M \leftarrow \mathrm{an\ } n\times n$ uniform copula
    \WHILE{converge}
    \STATE Choose a $2 \times 2$ submatrix  $\begin{pmatrix}
p_{ij}&p_{i, j+1}\\
p_{i+1, j}&p_{i+1, j+1}\\
\end{pmatrix}$ of $M$.
    \STATE Solve $$\frac{1}{p_{ij}+p_{i+1,j}+p_{i,j+1}+p_{i+1,j+1}}\log{\frac{(p_{ij}+\delta)(p_{i+1,j+1}+\delta)}{(p_{i+1,j}-\delta)(p_{i,j+1}-\delta)}} = r$$
    \STATE Update $\begin{pmatrix}
p_{ij}&p_{i,j+1}\\
p_{i+1,j}&p_{i+1,j+1}\\
\end{pmatrix}\to 
\begin{pmatrix}
p_{ij}+\delta&p_{i,j+1}-\delta\\
p_{i+1,j}-\delta&p_{i+1,j+1}+\delta\\
\end{pmatrix}$
    \ENDWHILE
    \RETURN $M$
    \end{algorithmic}
\end{algorithm}
The procedure 4 is to just solve a quadratic equation w.r.t. $\delta$ and thus computationally efficient. 
Figure \ref{fig:greedy-mick} shows the numerically obtained $30 \times 30$ MICK under $\mu = 0.5$. Note that Algorithm~\ref{alg1} requires in advance the value of pseudo log odds ratio instead of the rank correlation. The explicit relationship between pseudo log odds ratio and Kendall's $\tau$ remains unknown, however, it is possible to numerically specify the appropriate value of pseudo log odds ratio by binary search due to its monotonicity with respect to Kendall's $\tau$. See Appendix C. for more details on this method.

\begin{figure}[hbtp]
 \centering
 \includegraphics[keepaspectratio, scale=0.8]
      {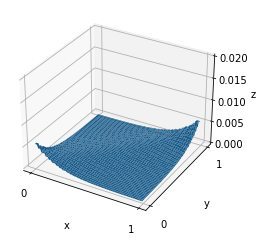}
 \caption{MICK, $n=30$, $\mu = 0.5$, pseudo log odds ratio = 2.9}
 \label{fig:greedy-mick}
\end{figure}

In addition, Algorithm~\ref{alg1} can also be applied to calculating MICS numerically when the procedure 4 is modified into 
$$\log{\frac{(p_{ij}+\delta)(p_{i+1,j+1}+\delta)}{(p_{i+1,j}-\delta)(p_{i,j+1}-\delta)}} = r',$$
where $r'$ denotes the value of log odds ratio. In practice, MICS becomes intractable numerically since the normalizing function is unknown. However, when we use our algorithm, this problem is overcome.


\subsection{Modelling dependence between two financial time series}

Aiming at the comparison between modelling using MICK and MICS, we demonstrate a plausible way to fit these copulas to real data in this subsection. Since these checkerboard copulas are equivalent to square matrices, they are numerically tractable.

Parameter estimation of these copulas can be done in an empirical manner using rank correlations, similar to the moment-based estimation. First, the data are preprocessed into order rank statistics. Then, the sample version of Kendall's $\tau$/Spearman's $\rho$ is calculated from observed data. Finally, we obtain MICK/MICS either by solving (P) or (RP) with numerical solvers, or by calculating the corresponding pseudo log odds ratio/log odds ratio (see Appendix C.) and then following Algorithm~\ref{alg1}.



We used log-return of daily stock price datasets, specifically the Dow Jones Average and S\&P500, which is displayed in Figure~\ref{fig:series}. All data were collected from Pandas-Datareader. The data length was 1635.
We observed four important statistics, the sample Kendall's tau $\tau$, the sample Spearman's rho $\rho$, empirically estimated lower tail dependence $\lambda_L^{u\%} = \mathrm{Pr}[F_2(X_2) < u/100\ |\ F_1(X_1) < u/100]\ (u=1,5)$, and empirically estimated upper tail dependence $\lambda_U^{u\%} = \mathrm{Pr}[F_2(X_2) > (100-u)/100\ |\ F_1(X_1) > (100-u)/100]\ (u=1,5)$, summarized as 
$$(\tau^{\mathrm{obs}}, \rho^{\mathrm{obs}}, \lambda_L^{5\%,{\mathrm{obs}}}, \lambda_U^{5\%,{\mathrm{obs}}},  \lambda_L^{1\%,{\mathrm{obs}}}, \lambda_U^{1\%,{\mathrm{obs}}}) = (0.802, 0.939, 0.827, 0.753, 0.812, 0.937).$$

\begin{figure}[ht]
 \centering
 \includegraphics[keepaspectratio, scale=0.25]
      {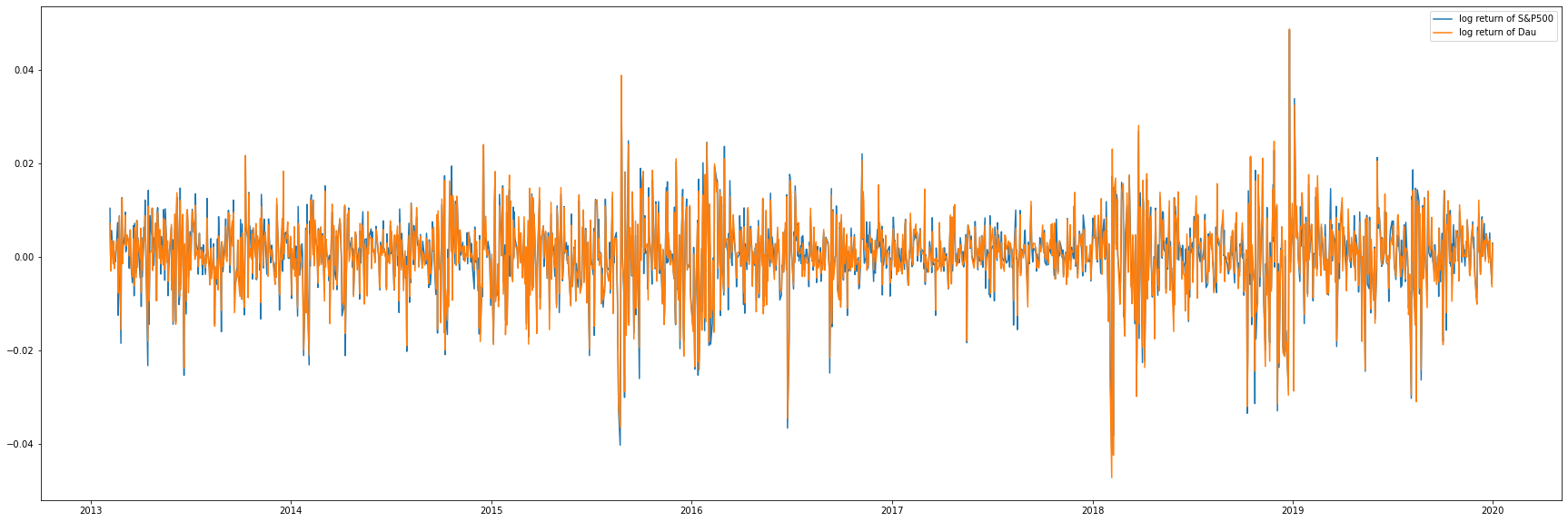}
 \caption{log return of Dow Jones Average and S\&P 500}
 \label{fig:series}
\end{figure}


\begin{figure}[ht]
 \centering
 \includegraphics[keepaspectratio, scale=0.5]
      {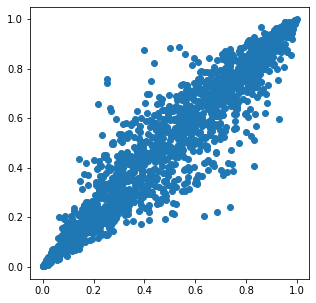}
 \caption{log return of Dow Jones Average vs S\&P 500 (preprocessed)}
 \label{fig:scatter-pre}
\end{figure}

As a result, the outputs from Algorithm~\ref{alg1} are depicted in Figure~\ref{fig:mick-sim} and Figure~\ref{fig:mics-sim} as scatter plots. The statistics of 150 samples randomly sampled from these checkerboard copulas are summarized in Table~\ref{tab:MICsimulation}.
Table~\ref{tab:MICsimulation} reveals that the simulation data generated with MICK exhibits a greater tail dependence coefficient than when using MICS. This suggests a more realistic capture of tail dependence in comparison. Nevertheless, it falls short of adequately representing the actual tail dependence in real-world observed data, which surpasses these values significantly. 
Several studies mention the advantages of using Kendall's tau in financial time series analysis~\cite{dehling_vogel_wendler_wied_2017}~\cite{huang2018}. However, our results suggest from an information-theoretic perspective that merely knowing the true value of Kendall's tau is not sufficient to explain the tail dependence observed in financial time series.
Refining the exact modeling of tail dependence within the framework of the minimum information copula (or entropy maximization) is a direction for future research.




\begin{figure}[ht]
  \begin{minipage}[b]{0.48\columnwidth}
    \centering
     \includegraphics[keepaspectratio, scale=0.5]
          {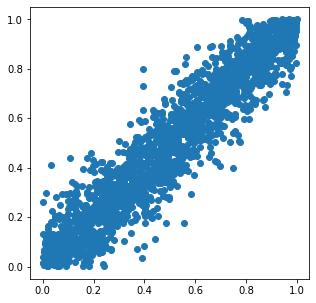}
     \caption{Simulation scattering of MICK ($\tau = 0.802$)}
     \label{fig:mick-sim}
  \end{minipage}
  \hspace{0.04\columnwidth} 
  \begin{minipage}[b]{0.48\columnwidth}
    \centering
     \includegraphics[keepaspectratio, scale=0.5]
          {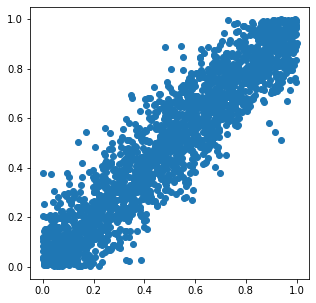}
     \caption{Simulation scattering of MICS ($\rho = 0.909$)}
     \label{fig:mics-sim}
  \end{minipage}
\end{figure}

\begin{table}[ht]
\caption{The statistics of observed and simulated data}
\label{tab:MICsimulation}
\centering
{\renewcommand\arraystretch{1.5}
\begin{tabular}{|c|ccc|} \hline
    &Observed&Simulated MICK&Simulated MICS\\ \hline
    $\tau$&0.802&0.801&0.774\\ \hline
    $\rho$&0.939& 0.949&0.939\\ \hline
    $\lambda_L^{5\%}$&0.827&0.467&0.367\\ \hline
    $\lambda_U^{5\%}$&0.753&0.471&0.369\\ \hline
    $\lambda_L^{1\%}$&0.812&0.112&0.086\\ \hline
    $\lambda_U^{1\%}$&0.937&0.115&0.086\\ \hline
\end{tabular}
}
\end{table}

\newpage

\section{Discussion}

\subsection{Relationship with Optimal Transport}
Finally, we point out the relationship between discrete version of minimum information copula and entropic optimal transport problem additionally.

In some articles, the equivalency between the minimum information copula and the optimal transport problem is mentioned~\cite{sei2021}. More specifically, the optimization problem for a minimum information copula is known to be equivalent to the optimal transport problem with entropy regularization. In fact, both Lagrangian coincides:

$$\mathcal{L}(P,\alpha,\beta) = \sum_{i=1}^n \sum_{j=1}^n (D_{ij}p_{ij}) + \frac{1}{\lambda} \sum_{i=1}^n \sum_{j=1}^n (p_{ij}\log{p_{ij}}-p_{ij})
+ \sum_{i=1}^n \alpha_i([P 1_n]_i - \frac{1}{n}) + \sum_{j=1}^n \beta_j([P^\top 1_n]_j - \frac{1}{n}),$$
where $\alpha \in \mathbb{R}^n$ and $\beta \in \mathbb{R}^n$ are the Lagrangian multipliers and $D_{ij}$ represents the cost of the transport from $i$ to $j$. MICS is also included in this case.

Analogously, the optimization problem (P) for MICK can also be interpreted as a new variant of optimal transport with entropy regularization controlled by parameter $\frac{1}{\lambda}$:
\begin{align*}
\min\ &1-\sum_{i=1}^n \sum_{j=1}^n \sum_{p=1}^n \sum_{q=1}^n p_{ij} p_{pq} \xi_{ip} \xi_{qj} + \frac{1}{\lambda} \sum_{i=1}^n \sum_{j=1}^n (p_{ij}\log{p_{ij}}-p_{ij})\\
&= \sum_{i=1}^n \sum_{j=1}^n \sum_{p=1}^n \sum_{q=1}^n p_{ij} p_{pq} v_{ijpq}
+ \frac{1}{\lambda} \sum_{i=1}^n \sum_{j=1}^n (p_{ij}\log{p_{ij}}-p_{ij}),
\end{align*}
$$\mathrm{s.t.}\ \sum_{i=1}^n p_{ij} = \frac{1}{n}, \sum_{j=1}^n p_{ij} = \frac{1}{n},$$
\noindent with its Lagrangian being
\begin{align*}
\mathcal{L} &= 1-\sum_{i=1}^n \sum_{j=1}^n \sum_{p=1}^n \sum_{q=1}^n p_{ij}p_{pq}\xi_{ip}\xi_{qj} + \frac{1}{\lambda} \sum_{i=1}^n \sum_{j=1}^n (p_{ij}\log{p_{ij}}-p_{ij}) + \sum_{i=1}^n \alpha_i([P 1_n]_i - \frac{1}{n}) + \sum_{j=1}^n \beta_j([P^\top 1_n]_j - \frac{1}{n}).
\end{align*}
\noindent As well as usual optimal transport problems, this problem can be interpreted as a bipartite graph matching: Let $G = (U \cup V, E)$ be a bipartite graph with $n^2$ transport edges, where $|U| = |V| = n$. Pick any couple of two edges. The cost $v_{ijpq} = -(\xi_{ip}-1)(\xi_{qj}-1)$ becomes $+1$ when the two edges $i\to j$ and $p \to q$ intersect, $0$ if either start or goal of two edges coincides, and $-1$ otherwise. The amount of mass is calculated as the multiplication of two transported masses. Here, we want to minimize the summation of each mass multiplied by each associated cost.

\subsection{Further extensions}
There has been an increasing interest in copula entropy theories as a solution to misspecification issue~\cite{bedford2014construction}. The concept of maximizing entropy corresponds to minimum information copulas since the famous Shannon entropy and information are closely linked. In this paper, we formulated a minimum information checkerboard copula under fixed Kendall's rank correlation and named it MICK. 

MICK is originally defined as the optimal solution of a non-convex programming. We confirmed that it can also be characterised in another way by the variant of odds ratio (pseudo log odds ratio). This result implies the global dependence is determined by the series of the local dependence. Taking advantage of this property, we constructed a quick algorithm to obtain MICK numerically even with a large grid size, enabling us to apply it to real data analysis. 

Since the concept of MICK is simplified, the extension of MICK is straightforward. 
First, the constraints to specify the copula density could be more general. In recent works, the expectation of a certain function, such as moments of the distribution, was only assumed as the constraints. We extended this to a second-order constraint in this work. Specifically, we focused on a single constraint that fixes Kendall's $\tau$ and also made a comparison with another important notion, Spearman's $\rho$. Recently, statistical notions that measure the dependence between multiple random variables in various domains, such as distance correlation. Extending the results to these is a future work of interest.

Furthermore, applying MICK to higher dimensions is of great interest. We developed MICK only for bivariate checkerboard copula in this paper. However, the results can be applied to three or more variables since the notion of information and Kendall's $\tau$ are not specific to bivariate cases. In three dimensional case for example, the given constraint should be fixing Kendall's $\tau$ for the trivariate checkerboard copula $P$: $p^\top (V \otimes V \otimes V) p$.

Last but not least, other than Shannon entropy considered in this paper, Tsallis entropy and Renyi entropy should also be discussed. It is possible that Tsallis entropy associated with power-law distribution naturally leads to copulas with heavier tail dependence, which is preferable for risk analysis in finance. 

\section*{Acknowledgements}
T. Sei is supported by JSPS KAKENHI Grant Numbers JP19K11865 and JP21K11781 and JST CREST Grant Number JPMJCR1763. 


\section*{Declaration of generative AI and AI-assisted technologies in the writing process}

During the preparation of this work the author(s) used ChatGPT (\url{https://chat.openai.com/})in order to proofread English. After using this tool/service, the author(s) reviewed and edited the content as needed and take(s) full responsibility for the content of the publication.

\bibliographystyle{plain}
\bibliography{retry}

\newpage

\appendix

\section{Proof for Theorem \ref{thm:uniqueness}.} \label{appendix:uniqueness}

\begin{proof}
To provide a sufficient condition for the uniqueness of MICK, 
we employ Gabrielsen~\cite{gabrielsen1986}, which states that it is sufficient to investigate the objective function at the stationary points under fairly weak assumptions.

Beforehand, we start from showing the monotonicity of Kendall's $\tau$ with respect to the coefficients of $T^{ij}$s.

\begin{lemma}\label{lem:anti-comonotone}
    Let $Q$ denote the $n\times n$ anti-comonotone checkerboard copula:
    $$Q = \begin{pmatrix}
        0&\cdots&0&\frac{1}{n}\\
        \vdots&\iddots&\iddots&0\\
        0&\iddots&\iddots&\vdots\\
        \frac{1}{n}&0&\cdots&0
    \end{pmatrix}.$$
    Then, for any checkerboard copula $P$, there exists non-negative $a_{ij}$s such that
    $$P = Q + \sum_{i=1}^{n-1} \sum_{j=1}^{n-1} a_{ij} T^{ij}.$$
\end{lemma}
\begin{proof}
    Define $\tilde{T}^{ijkl}$ as 
    $\tilde{T}^{ijkl} = \mathbf{e}_i \mathbf{e}_j^\top + \mathbf{e}_{k}\mathbf{e}_{l}^\top - \mathbf{e}_i\mathbf{e}_{l}^\top  - \mathbf{e}_{k}\mathbf{e}_{j}^\top \ (1 \leq i < k \leq n, 1\leq j < l \leq n).$ Then, it holds that
    $$\tilde{T}^{ijkl} = \sum_{i' = i}^{k-1} \sum_{j' = j}^{l-1} T^{i'j'}.$$
    Note that row sums and columns sums of $P-Q$ are all fixed to zero, and adding non-negative combination of various $-T^{ij}$s or even $-\tilde{T}^{ijkl}$s does not change any of these sums.
    Moreover, the anti-diagonal entries of $P-Q$ are all non-positive, and non-negative elsewhere:
    $$P-Q = \begin{pmatrix}
        +&\cdots&+&-\\
        \vdots&\iddots&\iddots&+\\
        +&-&\iddots&\vdots\\
        -&+&\cdots&+\\
    \end{pmatrix},$$
    where $+$ denotes a non-negative entry and $-$ denotes a non-positive entry. 
    Here we show that adding various $-T^{ij}$ to $P-Q$ repetitively leads to a zero matrix, which is equivalent to the statement.

    Let $\pi_{ij}$ denote the $(i,j)$-th entry of the matrix $P-Q$, where $1 \leq i,j \leq n$. First, we add $\sum_{i=1}^{n-1} \pi_{i1}(-\tilde{T}^{i,1,n,2}) + \sum_{j=3}^{n-1} \pi_{n,j}(-\tilde{T}^{n-1,2,n,j})$ to $P-Q$, aiming at eliminating positive entries in the first column and in the last row. The first term is intended to eliminate $\pi_{i1}\ (i=1,\cdots,n-1)$ and the second term to eliminate $\pi_{n,j}\ (j=3,\cdots,n)$. Note that $\pi_{n,1}$ becomes zero since the column sum of the first row is zero, and then  $\pi_{n,2}$ also becomes zero since the row sum of the last row is zero. Moreover, $\pi_{n-1,2}$ results in non-positive because $\pi_{i,2}\ (i=1,\dots,n-2)$ always stay non-negative and the column sum of the second row is always zero. Hence, the original matrix results in 
    $$P-Q+\sum_{i=1}^{n-1} \pi_{i1}(-\tilde{T}^{i,1,n,2}) + \sum_{j=3}^{n-1} \pi_{n,j}(-\tilde{T}^{n-1,2,n,j}) = \begin{pmatrix}
        0&+&\cdots&+&-\\
        \vdots&\vdots&\iddots&\iddots&+\\
        0&+&-&\iddots&\vdots\\
        \vdots&-&+&\cdots&+\\
        0&0&\cdots&\cdots&0\\
    \end{pmatrix},$$
    We observe that the $(n-1)\times(n-1)$ submatrix on upper right has identical signs with $P-Q$. This indicates that the repetition of this operation leads to zero matrix.
\end{proof}

From Lemma~\ref{lem:tau} and Lemma~\ref{lem:anti-comonotone}, the following lemma immediately follows.
\begin{lemma}\label{lem:monotone-tau-1}
    For any checkerboard copula $P (\neq Q)$, $\tau(\alpha P + (1-\alpha) Q)$, Kendall's $\tau$ of the convex combination of $Q$ and $P$, is strictly increasing with respect to $\alpha$, where $\alpha \in [0,1]$.
\end{lemma}
\begin{proof}
It suffices to show $\left. \frac{d}{d\alpha} \tau(\alpha P + (1-\alpha) Q)\ \right|_{\alpha = 1} \geq 0$, since $P$ can be chosen arbitrarily.
From Lemma~\ref{lem:anti-comonotone}, 
$$p-q = \sum_{i=1}^{n-1}\sum_{j=1}^{n-1} a_{ij}\mathbf{t}^{ij}, a_{ij}\geq 0$$
for any checkerboard copula $P$, where $p = \mathrm{vec}{(P)}$ and $q = \mathrm{vec}{(Q)}$.
Using Lemma~\ref{lem:tau}, the derivative of Kendall's $\tau$ with respect to $\alpha$ is 
\begin{align*}
    \left. \frac{d}{d\alpha} \tau(\alpha P + (1-\alpha) Q)\ \right|_{\alpha = 1} &= \left. 2\{\alpha(-p^\top W p + p^\top W q) + (1-\alpha)(q^\top W q - p^\top W q)\}\ \right|_{\alpha = 1}\\
    &= p^\top W (p-q)\\
    &= \sum_{i=1}^{n-1}\sum_{j=1}^{n-1} -2a_{ij} p^\top W\mathbf{t}^{ij}\\
    &> 0
\end{align*}
\end{proof}
\noindent From symmetry, the following lemma is also true.
\begin{lemma}\label{lem:monotone-tau-2}
    Let $R$ denote the comonotone checkerboard copula. 
    For any checkerboard copula $P (\neq R)$, $\tau(\alpha R + (1-\alpha) P)$, Kendall's $\tau$ of the convex combination of $P$ and $R$, is strictly increasing with respect to $\alpha$, where $\alpha \in [0,1]$.
\end{lemma}

Due to the monotonicity of $\tau$, we can find a star set that corresponds to the set of checkerboard copulas with the same value of Kendall's $\tau$.
\begin{lemma}\label{lem:connectivity}
    Let $Q$ denote the anti-comonotone checkerboard copula and $R$ denote the comonotone checkerboard copula. Let $\overline{QR}$ denote a line $\beta R + (1-\beta)Q$, where $\beta \in [0,1]$.
    For any checkerboard copula $P$, a curve through $P$, keeping the value of Kendall's $\tau$ the same, is connected to $\overline{QR}$.
\end{lemma}
\begin{proof}
    From Lemma~\ref{lem:monotone-tau-2}, $\tau(\gamma P + (1-\gamma) R)\geq\tau(P)$, where $\gamma \in [0,1]$. Therefore, from Lemma~\ref{lem:monotone-tau-1} and intermediate value theorem, there exists a unique $T$ such that $\tau(T) = \tau(P)$, where $T$ is a convex combination of $Q$ and $\gamma P + (1-\gamma) R$. For $\gamma \in [0,1]$, the orbit of $T$ forms a curve through $P$ where Kendall's $\tau$ stays constant, which is obviously connected to $\overline{QR}$.
\end{proof}

\begin{lemma}
Let 
$$S = \{p\in\mathbb{R}^{n^2}\ |\ p=\mathrm{vec}(P), \sum_{i=1}^n P_{ij} = \sum_{j=1}^n P_{ij} = \frac{1}{n}, \tau(P)=\mu\}$$
$$\bar{\Omega}' = \{x_{ij} = \frac{a_{ij}}{\sum_{k,l} a_{kl}}\ |\ P = Q + \sum_{i=1}^{n-1}\sum_{j=1}^{n-1} a_{ij} T^{ij}\},$$
$$\bar{\Omega} = \{x_{ij} = \frac{a_{ij}}{\sum_{k,l} a_{kl}}\ |\ P = Q + \sum_{i=1}^{n-1}\sum_{j=1}^{n-1} a_{ij} T^{ij}, \mathrm{vec}(P) \in S\} ,$$
$$\bar{\Omega} \subset \bar{\Omega}' \subset \mathbb{R}^{(n-1)^2-1}.$$
Then, \\
(i) $S$ has a one-to-one correspondence with $\bar{\Omega}$.\\
(ii) $\bar{\Omega}$ is a star domain.
\end{lemma}
\begin{proof}
The existence of the one-to-one correspondence follows from the monotonicity of $\tau$ with respect to $a_{ij}$ stated in Lemma~\ref{lem:monotone-tau-1}. We denote the projection from $S$ to $\bar{\Omega}$ as $x = x(p)$.
Here, we consider $x_0 \propto x(\mathrm{vec}(R))$, which clearly exists due to Lemma~\ref{lem:monotone-tau-1} and the intermediate value theorem. Then, $\bar{\Omega}$ is a star domain with its center being $x_0$ because 
for any $p \in S$, the curve passing through $p$ mentioned in the proof of Lemma~\ref{lem:connectivity} and the convex combination of $x_0$ and $x(p)$ is a one-to-one correspondence. 
\end{proof}




Hence, the objective function of (P) can be considered as a function on $\bar{\Omega}$ instead. Let $f(x), x \in\bar{\Omega}$ denote this function, then
the problem can be reformulated as minimizing $f(x)$ on $\bar{\Omega}$.
We wish to employ Theorem 2.1 of Gabrielsen~\cite{gabrielsen1986}, however, unfortunately it cannot be applied directly because $f(x) \neq 0, x \in \partial\bar{\Omega}$, where $\partial\bar{\Omega}$ denote the boundary of $\bar{\Omega}$. Here we introduce a different objective function $g(x)$ that does not alter the stationary points inherent in the problem (P), and satisfies $g(x) = 0, x \in \partial\bar{\Omega}$ simultaneously. Such $g$ exists in fact due to the following argument.

Let $\Omega \subset \mathbb{R}^{(n-1)^2-1}$ be an open set, whose closure is $\bar{\Omega}$. Then, $\Omega$ is a bounded star domain. Denote the boundary of $\Omega$ by $\partial \Omega$. 
For each $v \in \mathbb{S}= \{v \in \mathbb{R}^{(n-1)^2-1}\ |\ \|v\|= 1\}$, we can find a unique $t_*(v) >0$ such
that $x_0 + t_*(v)v \in \partial \Omega$ since $\Omega$ is a star.
Since $\mathbb{S}$ is compact and $t_*:\mathbb{S}\to\mathbb{R}$ is continuous, there exist twice differentiable continuous functions $t_1$ and $t_2$ on $\mathbb{S}$ such that $t_*(v)-\varepsilon < t_1(v) < t_2(v) < t_*(v)$. Let $\Omega_1 = \{x_0 + tv\ |\  v\in\mathbb{S}, 0 \leq t <t_1(v)\}$ and $\Omega_2 = \{x_0 + tv\ | \ v\in\mathbb{S}, 0 \leq t <t_2(v)\}$.
Since $x_0 (\propto x(\mathrm{vec}(R)))$ is a inner point of $\Omega$ because $R = Q + \sum_{i=1}^{n-1} \sum_{j=1}^{n-1} \frac{2\min{(i,j)}(1-\max{(i,j)})}{n^2} T^{ij}$, the following lemma holds.
\begin{lemma}\label{lem:taking-func-g}
There exists $\varepsilon > 0$ such that
$\frac{\mathrm{d}f(x_0 + tv)}{\mathrm{d}t} \geq 1$
for any $v \in \mathbb{S}$ and $t \in [t_*(v)-\varepsilon, t_*(v))$.

Then, we can find a function $g: \Omega \to \mathbb{R}$
such that\\
(i) $g(x) = f(x)$ for $x \in \Omega_1$,\\
(ii) $g$ is $C^2$ on $\Omega_2$,\\
(iii) $g(x) = M$ for $x\in \overline{\Omega}\backslash \Omega_2$, where $M= 1+\mathrm{sup}_{x\in\Omega} f(x)$, and\\
(iv) the stationary points of $g \in \Omega_2$ coincide with those of $f$.
\end{lemma}
\begin{proof}
    Let $\phi$ be a smooth increasing function on $\mathbb{R}$ such that $\phi(t) = 0$ for $t\leq0$ and $\phi(t) = 1$ for $t \geq 1$. Define a function $g$ by
    $$g(x_0 + tv) = 
    \begin{cases}
            f(x_0 + tv) - \phi\left(\frac{t-t_1(v)}{t_2(v)-t_1(v)}\right)\{M-f(x_0+tv)\} &(t \leq t_2)\\
            M&(t >t_2)
    \end{cases}
    $$
    Then the conditions (i) to (iv) are satisfied.
\end{proof}

Therefore, the problem reduces to proving the uniqueness of the optimal solution of the following problem:
$$\mathrm{Minimize}\ g(x) - M$$
$$\mathrm{s.t.}\ x \in \Omega_2.$$
\noindent Here, its feasible region is limited from $\overline{\Omega}$ to $\Omega_2$ because the stationary points of $g$ (or $f$) are in $\Omega_2$.
Due to Lemma~\ref{lem:taking-func-g}, $g(x)-M = 0$ on $x \in \partial\overline{\Omega_2}$. Therefore, Gabrielsen~\cite{gabrielsen1986} can be applied to this problem. Clearly, $\Omega_2$ is connected because it is a star. $\Omega_2$ is also non-empty and open, its closure $\overline{\Omega_2}$ is compact, and the objective function $g(x)-M$ is twice continuously differentiable on $\Omega_2$, we confirm the remaining conditions: Hessian is positive-definite at stationary points in $\Omega_2$.


Finally, we argue that Hessian matrix at every stationary points is positive-definite. 
Note that the local convexity on stationary points is invariant under the smooth coordinate transformation.
For the convenience of notation, the transpose symbol $\top$ in quadratic forms is omitted in the following proof when it is obvious.

Let $P_0$ be an optimal solution of $(\mathrm{P})$ and $P_t$ be a curve parameterized by $t\ (\in \mathbb{R})$ with a constant Kendall's $\tau$. In~\ref{appendix:uniqueness}, the symbol $'$ denotes the derivative in terms of $t$.
Since Kendall's $\tau$ is written as $\tau = 1-\mathrm{tr}(\Xi P \Xi P^\top) = 1-pWp$, $p_t$ satisfies the followings:
$$p_t W p_t = const.$$
$${p'_t} W p_t = 0.$$
$${p''_t} W p_t + {p'_t} W p'_t= 0,$$
where $p'_t$ and $p''_t$ are first and second derivative of $p_t$ w.r.t. $t$, respectively. 
The objective function $F$ along this curve is 
$$F_t = \sum_i \sum_j P_{t,ij} \log{P_{t,ij}}.$$
By taking derivatives with respect to the parameter $t$, we have
$$F'_t = \sum_i \sum_j P'_{t,ij} (\log{P_{t,ij}}+1),$$
$$F''_t = \sum_i \sum_j {P''}_{t,ij} (\log{P_{t,ij}}+1) + \frac{{P'}_{t,ij}^2}{P_{t,ij}}.$$
The Lagrangian is written as 
$$\mathcal{L} = \sum_{i=1}^n\sum_{j=1}^n P_{ij}\log{P_{ij}} - \lambda(1-\mathrm{tr}(\Xi P\Xi P^{\top}) - \mu) - \sum \alpha_j( \sum_{i=1}^n P_{ij} - \frac{1}{n})- \sum \beta_i ( \sum_{j=1}^n P_{ij} - \frac{1}{n}).$$
\noindent Since $P_0$ is the optimal solution, $P_0$ should satisfy the following stationary condition:
$$\frac{\partial \mathcal{L}}{\partial P_{ij}} = \log{P_{ij}}+1 + \lambda (W p)_{ij} - \alpha_j - \beta_i = 0.$$
Therefore, with $t=0$,
$$F''_t = -\sum_i \sum_j {P''}_{t,ij} \lambda (W p_t)_{n(i-1)+j} + \sum_i \sum_j {P''}_{t,ij}(\alpha_j+\beta_i) + \sum_i \sum_j \frac{{P'}_{t,ij}^2}{P_{t,ij}} $$

For any $t$, row sums and columns sum are constants:$\sum_i P_{ij} = \frac{1}{n}, \sum_j P_{ij} = \frac{1}{n}$. Hence, by taking derivatives with $t$, $\sum_i {P''}_{t,ij} = 0, \sum_j {P''}_{t,ij} = 0. $ Therefore, 
\begin{align*}
F'' &= -\sum_i \sum_j \lambda ({{p_t}''}_{n(i-1)+j} (W p_t)_{n(i-1)+j}) + \sum_i \sum_j \frac{{P'}_{t,ij}^2}{P_{t,ij}} \\
&= -\lambda{p_t}'' W p_t + {p_t}' \mathrm{Diag}(\frac{1}{p_{t,n(i-1)+j}}) {p_t}'\\
&= \lambda{p_t}' W {p_t}' + {p_t}' \mathrm{Diag}(\frac{1}{p_{t,ij}}) {p_t}'\\
&= {p_t}' (\mathrm{Diag}(\frac{1}{p_{t,ij}})+\lambda W) {p_t}'
\end{align*}
when $t=0$.

\noindent Now, we consider changing the basis by 
$$p = (A \otimes A) q + \frac{1}{n^2}1_{n^2},$$
where $p \in [0,\frac{1}{n}]^{n^2}, q \in \mathbb{R}^{(n-1)^2}$.
Then, by incorporating $p'_t = (A \otimes A) q'_t$, we have
$$F'' = (q_t')^\top (A \otimes A)^\top \mathrm{Diag}(\frac{1}{p_{0,ij}}) (A \otimes A) q_t' + \lambda (q_t')^\top (A \otimes A)^\top W (A \otimes A) q_t'.$$

Define $D_1 = (A \otimes A)^\top \mathrm{Diag}(\frac{1}{p_{0,ij}}) (A \otimes A)$ and $D_2 = (A \otimes A)^\top W (A \otimes A)$. From tedious calculation, $D_2 = -M \otimes M$, where $$M = \begin{pmatrix}
0&-1&\dots&0\\
1&0&\ddots&\vdots\\
\vdots&\ddots&\ddots&-1\\
0&\dots&1&0\\
\end{pmatrix} = \mathrm{Toeplitz}(1,0,-1) \in \mathbb{R}^{n \times n}.$$ 
The problem is to examine whether $D_1+\lambda D_2$ is positive-definite. Here, we seek to evaluate $|\frac{\lambda v^\top D_2 v}{v^\top D_1 v}|$. It is easy to see that $|\frac{\lambda v^\top D_2 v}{v^\top D_1 v}| < 1$ for any $v \in \mathbb{R}^{n^2}$ is a sufficient condition for the positive-definiteness of $D_1+\lambda D_2$. First, we evaluate $|\frac{v^\top D_2 v}{v^\top D_1 v}|$ as follows. Note that $D_1$ is positive-definite thus $v^\top D_1 v$ is always positive.
\begin{align*}
    {\max_{v} \frac{v^\top D_2 v}{v^\top D_1 v}}
    &=  {\max_{v} \frac{v^\top D_2 v}{v^\top LL^\top v}}\\
    &=  {\max_{x, x=L^\top v} \frac{x^\top L^\dagger (-M \otimes M) (L^\dagger)^\top x}{x^\top x}}\\
    &\leq {\lambda_{\max} (L^\dagger (-M \otimes M) (L^\dagger)^\top)},
\end{align*}
where $$L = (A \otimes A)^\top \mathrm{Diag}(\frac{1}{\sqrt{p_0}})$$
and
$$L^\dagger = \mathrm{Diag}(\sqrt{p_0}) (A^\dagger \otimes A^\dagger)^\top.$$
Similarly, $\min_v \frac{v^\top D_2 v}{v^\top D_1 v} \geq  \lambda_{\min} (L^\dagger (-M \otimes M) (L^\dagger)^\top).$ Therefore, $|\frac{v^\top D_2 v}{v^\top D_1 v}|$ is upper bound by the eigenvalue of $(L^\dagger (-M \otimes M) (L^\dagger)^\top)$ with maximum absolute value. 


Now, the target matrix of interest can be transformed as follows:
\begin{align*}
    L^\dagger (-M \otimes M) (L^\dagger)^\top &= \mathrm{Diag}(\sqrt{p_0}) (A^\dagger \otimes A^\dagger)^\top (-M \otimes M) (A^\dagger \otimes A^\dagger) \mathrm{Diag}(\sqrt{p_0})\\
    &= - \mathrm{Diag}(\sqrt{p_0})(X \otimes X) \mathrm{Diag}(\sqrt{p_0}),
\end{align*}
where $X = (A^\dagger)^\top M A^\dagger$. Through tedious algebraic calculations, we have
$$X = \mathrm{Toeplitz}\left(-\frac{n-2}{n}, -\frac{n-4}{n}, \dots,\frac{n-2}{n}, 0, -\frac{n-2}{n},-\frac{n-4}{n},\dots, \frac{n-2}{n}\right),$$
\noindent $\mathrm{Diag}(\sqrt{p_0})(X \otimes X) \mathrm{Diag}(\sqrt{p_0})$ is a block matrix, and its $(i,j)$-th block is 
$$X_{ij}\mathrm{Diag}(\sqrt{p_{i1}}, \dots ,\sqrt{p_{in}}) X \mathrm{Diag}(\sqrt{p_{j1}}, \dots ,\sqrt{p_{jn}}).$$
Hence, the $(k,l)$-th entry of $(i,j)$-th block ($k,l,i,j = 1,\dots,n)$ is written as 
$$X_{ij}X_{kl}\sqrt{p_{ik}}\sqrt{p_{jl}}.$$


The range of eigenvalues of $-\mathrm{Diag}(\sqrt{p_0})(X \otimes X) \mathrm{Diag}(\sqrt{p_0})$ is specified by Gershgorin's theorem for block operetors. The radius of the Gershgorin circle is determined by absolute column sums in a single block.

\begin{theorem}[Gershgorin's Theorem for Block Matrices(Tretter, 2008. section 1.13 ~\cite{tretter2008spectral})]
Let $n \in \mathbb{N}$ and $\mathcal{A} = (A_{ij})$ a block matrix with symmetric diagonal entries $A_{ii}$. If we define
$$\mathcal{G}_i = \sigma(A_{ii}) \cup \{\lambda\ :\ |\lambda - \sigma(A_{ii})| \leq \sum_{j=1, j\neq i}^n \|A_{ij}\| \},$$
for $i=1,\dots,n$, where $\|\cdot\|$ denotes a submultiplicative norm, then 
$$\sigma(\mathcal{A}) \subset \bigcup_{i=1}^n \mathcal{G}_i.$$
\end{theorem}


Note that all diagonal blocks of the target matrix are zero matrices, thus $\sigma(A_{ii})=0$ and the centers of each circle are zero.
For any block $B$ of the target matrix, the Frobenius norm of $B$ is calculated as 
\begin{align*}
\|B\|_F &= |X_{ij}| \sqrt{\sum_k \sum_l (X_{kl}\sqrt{p_{ik}}\sqrt{p_{jl}})^2}\\
&\leq |X_{ij}| \sqrt{\sum_k \sum_l p_{ik}p_{jl}}\\
&= \frac{1}{n} |X_{ij}|,\\
\sum_{j \neq i} \|B_{ij}\|_F &\leq \sum_{j=1}^n \frac{1}{n} |X_{ij}|\\
&\leq \frac{1}{n} \frac{n-2}{2}\\
&< \frac{1}{2}.
\end{align*}

\noindent Since Frobenius norm is submultiplicative, we may apply the Gershgorin's Theorem. Therefore, the eigenvalues of $\mathrm{Diag}(\sqrt{p_0})(X \otimes X) \mathrm{Diag}(\sqrt{p_0})$ (or $-\mathrm{Diag}(\sqrt{p_0})(X \otimes X) \mathrm{Diag}(\sqrt{p_0})$ as well) are between $-\frac{1}{2}$ and $\frac{1}{2}$, which gives us $|\frac{v^\top D_2 v}{v^\top D_1 v}| \leq \frac{1}{2}$.

We conclude that under $|\lambda| < 2$, $|\frac{\lambda v^\top D_2 v}{v^\top D_1 v}| =  |\lambda| |\frac{v^\top D_2 v}{v^\top D_1 v}| \leq 1$ thus $D_1+\lambda D_2$ is always positive-definite, meaning that $F''$ is positive and $F$ is locally strictly convex around every stationary points of (P). 
By employing Theorem 2.1 in Gabrielsen~\cite{gabrielsen1986}, the stationary point of (P) exists uniquely and the optimal solution as well.








\end{proof}

\section{Relationship between pseudo log odds ratio and Kendall's $\tau$ / log odds ratio and Spearman's $\rho$}

Figure~\ref{fig:relationship-mick} shows the correspondence between pseudo log odds ratio and Kendall's $\tau$ of MICK. Similarly, Figure~\ref{fig:relationship-mics} shows the correspondence between log odds ratio and Spearman's $\rho$ of MICS. Note that this one-to-one correspondence has been proven to be true for MICS~\cite{sei2021} but has not been established yet for MICK.
The monotonicity in Figure \ref{fig:relationship-mics} is consistent with the following theorem:
\begin{theorem}[Theorem 4.3, Sei~\cite{sei2021} ]
    $$\mu(\theta) = \frac{\partial \psi}{\partial \theta},$$
    where $\psi$ is a convex function. Hence, $\mu(\theta)$ is monotone:
    $$\frac{\partial \mu(\theta)}{\partial \theta} = \frac{\partial^2 \psi}{\partial \theta^2} \geq 0.$$
\end{theorem}

Additionally, Table \ref{tab:mick-param} and Table \ref{tab:mics-param} are the correspondence tables showing the relationship between Kendall's $\tau$/Spearman's $\rho$ and pseudo log odds ratio/log odds ratio depicted in Figure~\ref{fig:relationship-mick} and Figure~\ref{fig:relationship-mics}, respectively. 
In practice, these tables are useful for obtaining MICK/MICS numerically using Algorithm~\ref{alg1}. They interpret rank correlations into ratios based on the one-to-one correspondence exhibited in Figure~\ref{fig:relationship-mick} and Figure~\ref{fig:relationship-mics}. Hence, although ratios cannot be calculated directly from observed samples, the values of ratios required as input for Algorithm~\ref{alg1} can be obtained through the sample version of rank correlations.

\begin{figure}[tbp]
  \begin{minipage}[b]{0.48\columnwidth}
    \centering
     \includegraphics[keepaspectratio, scale=0.6]
          {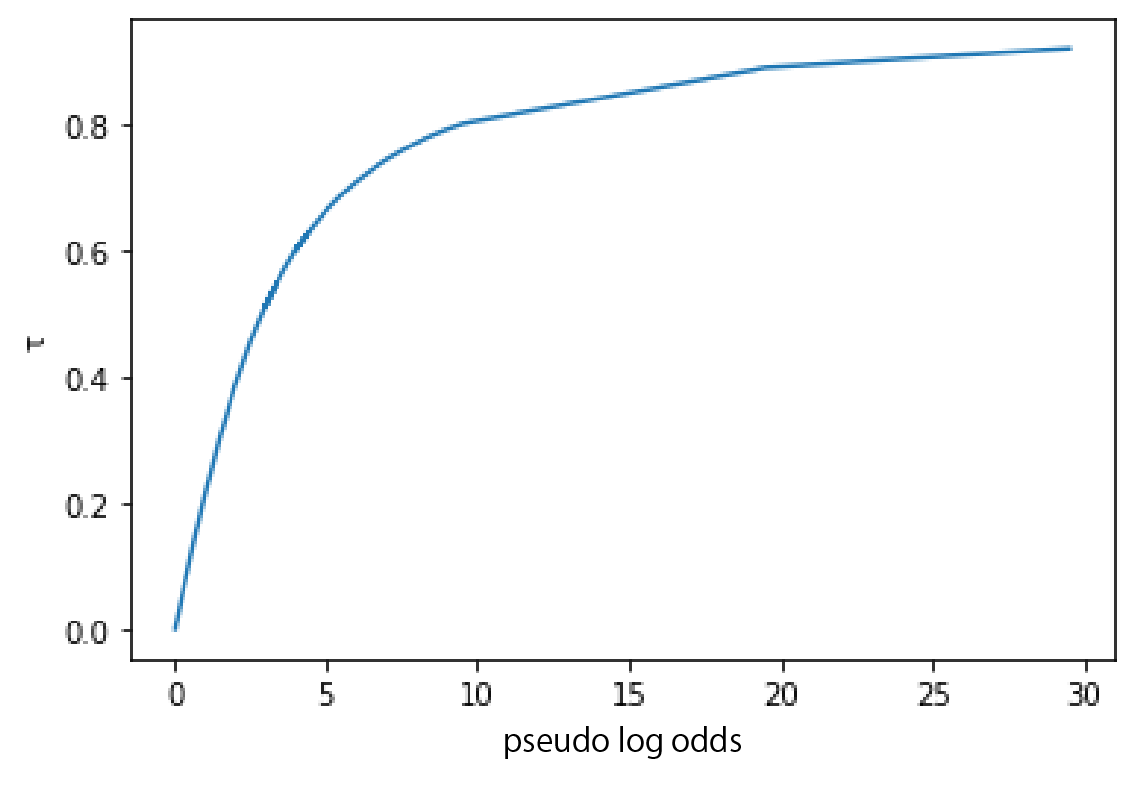}
     \caption{pseudo log odds ratio vs $\tau$ for $30 \times 30$ MICK}
     \label{fig:relationship-mick}
  \end{minipage}
  \hspace{0.04\columnwidth} 
  \begin{minipage}[b]{0.48\columnwidth}
    \centering
     \includegraphics[keepaspectratio, scale=0.5]
          {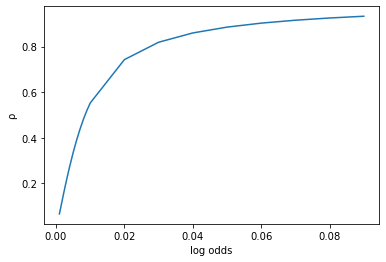}
     \caption{log odds ratio vs $\rho$ for $30 \times 30$ MICS}
     \label{fig:relationship-mics}
  \end{minipage}
\end{figure}

\begin{table}[]
\begin{minipage}[b]{0.45\columnwidth}
    \centering
    \begin{tabular}{c|c|c|c}
      pseudo log odds ratio&$\rho$&$\tau$&information\\ \hline
      0.300&0.091&0.060&-6.798\\
      0.500&0.156&0.104&-6.789\\
      1.000&0.309&0.208&-6.752\\
      2.000&0.552&0.384&-6.624\\
      3.000&0.707&0.511&-6.468\\
      4.000&0.801&0.599&-6.315\\
      5.000&0.858&0.662&-6.174\\
      6.000&0.894&0.709&-6.048\\
      7.000&0.918&0.744&-5.934\\
      8.000&0.934&0.771&-5.832\\
      9.000&0.945&0.792&-5.741
    \end{tabular}
    \caption{Parameters of $30 \times 30$ MICK}
    \label{tab:mick-param}
\end{minipage}
\hfill
\begin{minipage}[b]{0.45\columnwidth}
    \centering
    \begin{tabular}{c|c|c|c}
      log odds ratio&$\rho$&$\tau$&information\\ \hline
0.001&0.066&0.044&-6.800\\
0.002&0.139&0.093&-6.792\\
0.003&0.209&0.140&-6.780\\
0.004&0.274&0.184&-6.763\\
0.005&0.334&0.225&-6.743\\
0.006&0.388&0.262&-6.721\\
0.007&0.437&0.296&-6.698\\
0.008&0.480&0.327&-6.674\\
0.009&0.518&0.355&-6.650\\
0.01&0.552&0.380&-6.626\\
0.02&0.742&0.534&-6.426\\
0.03&0.819&0.609&-6.287\\
0.04&0.859&0.656&-6.181\\
0.05&0.885&0.689&-6.096\\
0.06&0.902&0.713&-6.024\\
0.07&0.915&0.733&-5.962\\
0.08&0.925&0.749&-5.907\\
0.09&0.933&0.762&-5.859\\
    \end{tabular}
    \caption{Parameters of $30 \times 30$ MICS}
    \label{tab:mics-param}
\end{minipage}
\end{table}

\end{document}